%% file: main.tex
\documentclass[11pt,reqno,a4paper,oneside]{amsart}	

\usepackage{amsmath, amssymb, amsthm, amstext}		
\usepackage{txfonts}								
\usepackage[mathletters]{ucs}
\usepackage[utf8x]{inputenc}
\usepackage[english]{babel}
\usepackage{harvard}								
\usepackage[T1]{fontenc}
\usepackage{ae,aecompl}
\usepackage{enumerate}

\usepackage{appendix}
\usepackage{lmodern}
\usepackage{epsfig,color}
\usepackage[arrow, matrix, curve]{xy} 				
\usepackage{tensor} 								
\usepackage{cancel}								

\usepackage{color}								
\usepackage{mathrsfs}
\usepackage{bbm}
\usepackage{comsym}
\usepackage{graphicx}

\usepackage{babelbib}							

\numberwithin{equation}{section}

\newtheorem{lem}{Lemma}[section]
\newtheorem{prop}[lem]{Proposition} 
\newtheorem{thm}[lem]{Theorem}
\newtheorem{cor}[lem]{Corollary}
\newtheorem{conj}[lem]{Conjecture} 
\newtheorem{summ}[lem]{Summary}

\theoremstyle{definition}

\newtheorem{ex}[lem]{Example}

\newtheorem*{thm*}{Theorem}

\usepackage[textwidth=5in,textheight=8.25in,paperheight=11in]{geometry}
\usepackage{setspace}
\allowdisplaybreaks[1]  
\usepackage{hyperref}
\hypersetup{%
  pdftitle = {Twistors and GR},
  pdfauthor = {Norman Metzner},
  bookmarksnumbered=true,
  bookmarksdepth =1,
  pdfdisplaydoctitle=true,
  linkbordercolor={0 0 1}
}

\input{headings}

\title[Twistor Theory of Higher-Dimensional Black Holes II]{Twistor Construction of Higher-Dimensional Black Holes --- Part~II: Examples}

\author{Paul Tod$\mbox{}^{\dagger}$}
\thanks{$\mbox{}^{\dagger}$Mathematical Institute, University of Oxford, 24-29 St\,Giles', OX1 3LB Oxford, UK, and St\,John's College, OX1 3JP, Oxford, UK}

\author{Norman Metzner$\mbox{}^{\ddagger}$}
\thanks{$\mbox{}^{\ddagger}$Mathematical Institute, University of Oxford, 24-29 St\,Giles', OX1 3LB Oxford, UK, and St\,John's College, OX1 3JP, Oxford, UK}

\author{Lionel Mason$\mbox{}^{+}$}
\thanks{$\mbox{}^{+}$Mathematical Institute, University of Oxford, 24-29 St\,Giles', OX1 3LB Oxford, UK, and St\,Peter's College, New Inn Hall Street, OX1 2DL, Oxford, UK}

\begin{document}


\addtolength{\jot}{0.1cm}			

\begin{abstract}
We apply the twistor construction for higher-dimensional black holes to known examples in five space-time dimensions. First the patching matrices are calculated from the explicit metric for these examples. Then an ansatz 
is proposed for obtaining the patching matrix instead from the data of rod structure and angular momenta. The ansatz is tested on examples with up to three nuts, and these are shown 
to give flat space, the Myers-Perry solution and the black ring, as expected. Rules for the transition between different adaptations of the patching matrix and for the elimination of conical singularities are 
developed and seen to work.
\end{abstract}

\maketitle   
\tableofcontents

\input{Introduction}
\input{Pexamples}
\input{TheConverse}
\input{Outlook}

\hypersetup{
	bookmarksdepth=0
}

\section*{Acknowledement}
We are grateful to Piotr Chru\'{s}ciel and Nicholas Woodhouse for helpful discussions and sharing their thoughts on some ideas. LM is supported by a Leverhulme Fellowship and EPSRC grant EP/J019518/1. NM was supported by a PhD studentship from the German National Academic Foundation (Studienstiftung des deutschen Volkes), a Lamb and Flag Scholarship from St\,John's College Oxford and the EPSRC studentship MATH0809. 

\nocite{Ward:1983yg}
\bibliographystyle{jphysicsB}						
\bibliography{/Users/norman/mathematics/Papers/library} 							
\end{document}

%% file: Introduction.tex
\section{Introduction}
This paper follows Part~I in which basic results were established and we shall freely use those 
results and the terminology established there. Thus in five 
dimensions the twistor construction for black hole space-times can be summarized as follows.

\begin{summ} \label{thm:summary}
There exists a one-to-one correspondence between five-dimensional stationary and axisymmetric space-times and rank-3 bundles $E→\mathcal{R}$ over reduced twistor space $\mathcal{R}$, where $\mathcal{R}$ consists of two Riemann spheres identified over a certain region. $E$ can be encoded in a symmetric meromorphic $3\times 3$ matrix $P(z)$ as follows.

If $J$ is the matrix of inner products of Killing vectors, we define the Ernst potential adapted to a particular rod as the matrix
\begin{equation*} 
\renewcommand{\arraystretch}{1.4}
J'=\frac{1}{\det \skew{7}{\tilde}{A}} \left(\begin{array}{cc}\hphantom{-}1 & -χ^{\mathrm{t}} \\-χ & \det \skew{7}{\tilde}{A} \cdot \skew{7}{\tilde}{A} + χχ^{\mathrm{t}}\end{array}\right),
\end{equation*}
where $\skew{7}{\tilde}{A}$ is obtained from $J$ by cancelling the appropriate row and column, and $χ=(χ_{1},χ_{2})$ are the twist potentials.

The bundle $E→\mathcal{R}$ is characterized by the \emph{twistor data} which for an axis-regular Ernst potential consists only of the patching matrix $P$. The patching matrix is an analytic continuation of the Ernst potential, that is $P(z)=J'(0,z)$ where $J'$ is non-singular for $r→0$. 
\end{summ}

Note that the axis-regularity fixes the three integers which are initially part of the twistor data. For the bundle corresponding to $J$ itself the integers are $p_{0}=1$, $p_{1}=p_{2}=0$ and for the bundle corresponding to the Ernst potential $J'$ they are $p_{0}=p_{1}=p_{2}=0$. 

Moreover, we have seen in Part~I that the rod structure of the space-time is coded into the poles and residues of the patching matrix.
\begin{prop}
A patching matrix $P$ has real singularities, that is points $z\in \mathbb{R}$ where an entry of $P$ has a singularity, at most at the nuts of the rod structure and these real singularities are simple poles of $P$.
\end{prop}
This proposition requires horizons to be nondegenerate. From the definition of $J'$ it is clear that there are different patching matrices adapted to different rods. 
In Section~\ref{sec:Pexamples} we calculate the patching matrix on the top-end rod for the Myers-Perry solution  
and the black ring solution; for flat space and the five-dimensional Schwarzschild space-time the patching matrix is easily computed on all rods.

The converse direction will be studied in Section~\ref{sec:converse}, that is we present an ansatz for constructing the patching matrix given the data of angular momenta and the rod structure. 
By Summary~\ref{thm:summary}, knowing $P$ is equivalent to determining the space-time metric. Not all rod structures lead to solutions. 
In order to determine the genuinely free parameters in $P$, it is necessary to understand how $P$ behaves when 
changing from the adaptation to one rod to the one for an adjacent rod, a process we call switching. In Theorem~6.5 of Part ~I we have seen how to switch around the nut at infinity and 
in Section~\ref{sec:converse} we study this for an arbitrary nut. Finally, we show how to eliminate conical singularities and apply this to the black ring.

%% file: Pexamples.tex
\section{Patching Matrix for Relevant Examples} \label{sec:Pexamples}
In order to find the patching matrix we first need to know the metric of our space-time in the ${\sigma}$-model form, that is we have to calculate $J(r,z)$. After that the Ernst potential, 
respectively the patching matrix, can be computed which mainly means determining the twist potentials on the axis. The easiest example to start with is flat space in five dimensions. 

\subsection{Five-Dimensional Minkowski Space} \label{subsec:exMinkrodstr}
The real five-dimensional Minkowski space is the manifold $\mathbb{R}^{5}$ with the metric that, in double polar coordinates, takes the form
\begin{equation*}
\drm s^{2} = -\drm x_{0}^{2}+\drm R_{1}^{2} + R_{1}^{2} \,\drm {\varphi}^{2}+\drm R_{2}^{2} + R_{2}^{2} \,\drm {\psi}^{2}.
\end{equation*}
The rotational Killing vector fields are $X_{1}=∂_{\varphi}$, $X_{2}=∂_{\psi}$. To obtain the ${\sigma}$-model coordinates we introduce $z$, $r$ by 
\begin{equation*}
z+\irm r = \frac{1}{2} \left(R_{1}+\irm R_{2}\right)^{2}. 
\end{equation*}
Then
\begin{align} \label{eq:JMink}
J(r,z) & =
\renewcommand\arraystretch{1.5}
\left(\begin{array}{ccc}-1 & 0 & 0 \\
\hphantom{-}0 & z+\sqrt{r^{2}+z^{2}} & 0 \\
\hphantom{-}0 & 0 & -z+\sqrt{r^{2}+z^{2}}
\end{array}\right)\\
\text{and} \quad {\erm}^{2\nu} & = \frac{1}{2\sqrt{r^{2}+z^{2}}}. 
\end{align}
Since $\dim(\ker J(0,z))>1$ only for $z=0$, we can read off that the metric admits two semi-infinite rods, namely $(-\infty,0)$ and $(0,\infty)$. Because $J$ is diagonal we have for the 
Killing 1-forms $\theta_{I}=g(X_{I},⋅)$ that ${\theta}_{I}\wedge \drm {\theta}_{I}=0$. It follows that the twist potentials are zero without loss of generality and thereby we obtain the patching matrix as 
\begin{equation} \label{eq:PMink}
P_{±}(z)=\mathrm{diag}\, \left(\mp \frac{1}{2z},-1,\pm  2z\right),
\end{equation}
where the upper sign combination is for $P$ adapted to $z>0$, and the lower one for $z<0$.

\subsection{Twist Potentials on the Axis}

As part of the algorithm for obtaining $P(z)$ from the metric we need to calculate the twist potentials just on the axis. Explicit expressions for twist potentials have been 
obtained for example in \cite{Tomizawa:2004aa} and \cite{Tomizawa:2009aa}, but these are not in Weyl coordinates which we need here. 
Therefore it is simpler to rederive some results, not only for completeness but also for providing a way of calculating the twist potentials on the axis for other space-times where they are not yet in the literature.

First we derive general formulae. Assume that the metric takes the form
\begin{align*}
\drm s^{2} & = J_{00} \,\drm t^{2} + 2 J_{01}\, \drm t \drm {\varphi} + 2 J_{02}\, \drm t \drm {\psi} + J_{11}\, \drm {\varphi} ^{2} + 2 J_{12}\,\drm {\varphi} \drm {\psi} \\
& \hspace{0.4cm} + J_{22}\, \drm {\psi}^{2} + \erm^{2{\nu}} \left(\drm r^{2}+\drm z^{2}\right),
\end{align*}
and rewrite it as 
\begin{align*}
\drm s^{2} & = -F^{2} \left(\drm t + {\omega}_{1}\, \drm {\varphi} + {\omega}_{2}\,\drm {\psi} \right)^{2} + G^{2} \left(\drm {\psi}+ {\Omega} \, \drm {\varphi}\right)^{2} \\ 
& \hspace{0.4cm} + H^{2}\, \drm {\varphi}^{2} + \erm^{2{\nu}} \left(\drm r^{2}+\drm z^{2}\right),
\end{align*}
with
\begin{align*}
& F^{2} = - J_{00}, \quad  -F^{2} {\omega}_{1} = J_{01}, \quad -F^{2} {\omega}_{2} = J_{02}, \\
- & F^{2}{\omega}_{2}^{2}+G^{2} = J_{22}^{\vphantom{1}}, \quad  -F^{2} {\omega}_{1} {\omega}_{2}+G^{2}{\Omega} = J_{12}, \\
- & F^{2}{\omega}_{1}^{2}+G^{2}{\Omega}^{2}+H^{2} = J_{11}^{\vphantom{1}}.
\end{align*}
The latter form has been chosen to facilitate calculating $P$ adapted to part of the axis where $z\to \infty $ and $\partial_{\varphi}=0$. In terms of the orthonormal frame
\begin{align*}
{\theta}^{0} & = F \left(\drm t + {\omega}_{1}\, \drm {\varphi} + {\omega}_{2}\,\drm {\psi} \right), \quad {\theta}^{1} = G \left(\drm {\psi}+ {\Omega} \, \drm {\varphi}\right), \\
{\theta}^{2} & = H \, \drm {\varphi}, \quad {\theta}^{3} = \erm ^{{\nu}} \, \drm r, \quad {\theta}^{4} = \erm ^{{\nu}} \, \drm z,
\end{align*}
the Killing 1-forms take the form
\begin{align*}
\frac{∂}{\partial t} & \to  T = - F\,{\theta}^{0} = - F^{2} \left(\drm t + {\omega}_{1}\, \drm {\varphi} + {\omega}_{2}\,\drm {\psi} \right),\\
\frac{∂}{\partial \psi} & \to  Ψ = G\,{\theta}^{1}-F{\omega}_{2}\,{\theta}^{0}.
\end{align*}
Using $\drm {\varphi} = H^{-1}\, {\theta}^{2}$, $\drm {\psi} = G^{-1}\,{\theta}^{1}-{\Omega}H^{-1}\,{\theta}^{2}$ this yields for the first twist potential
\begin{align*}
\drm {\chi}_{1} & = * \left(T\wedge Ψ\wedge \drm T\right) \\
 & = - \frac{F^{3}G}{H} *\left({\theta}^{0}\wedge {\theta}^{1}\wedge {\theta}^{2}\wedge (\drm {\omega}_{1}-{\Omega} \,\drm {\omega}_{2} )\right),
\end{align*}
and for the second
\begin{align*} 
 \drm {\chi}_{2} & = * \left(T\wedge Ψ\wedge \drm Ψ\right)\\
 & = \frac{FG}{H} *\left({\theta}^{0}\wedge {\theta}^{1}\wedge {\theta}^{2}\wedge (G^{2}\,\drm {\Omega}-F^{2}{\omega}_{2}\, \drm {\omega}_{1}-F^{2}{\omega}_{2}{\Omega} \,\drm {\omega}_{2} )\right).
\end{align*}
Since $J=J(r,z)$ all the functions depend only on $r$, $z$, hence so do ${\chi}_{i}$ and ${\omega}_{i}$. Then the total derivatives are $\drm {\chi}_{i} = \partial _{r} {\chi}_{i}\, \drm r+ \partial _{z} {\chi}_{i}\, \drm z$ and analogous for ${\omega}_{i}$. Furthermore, noting that $\drm r = \erm ^{-{\nu}} \, {\theta}^{3}$, $\drm z = \erm ^{-{\nu}} \, {\theta}^{4}$ the above equations read
\begin{align*}
\drm {\chi}_{1} & = - {\epsilon} \, \frac{F^{3}G}{H} \Big((\partial _{r} {\omega}_{1}-{\Omega} \, \partial _{r} {\omega}_{2})\drm z-(\partial _{z} {\omega}_{1}-{\Omega} \, \partial _{z} {\omega}_{2}) \drm r\Big)\\
\Rightarrow  \partial _{z}{\chi}_{1} & = -{\epsilon} \, \frac{F^{3}G}{H}\Big(\partial _{r} {\omega}_{1}-{\Omega} \, \partial _{r} {\omega}_{2}\Big),
\end{align*}
and
\begin{align*}
\drm {\chi}_{2} & = {\epsilon}\, \frac{FG}{H} \Bigg(\Big(G^{2} \partial _{r} {\Omega}-F^{2}{\omega}_{2} \partial _{r} {\omega}_{1}-F^{2}{\omega}_{2}{\Omega} \partial _{r} {\omega}_{2}\Big) \drm z \\
 & \hphantom{=} - \Big(G^{2} \partial _{r} {\Omega}-F^{2}{\omega}_{2} \partial _{r} {\omega}_{1}-F^{2}{\omega}_{2}{\Omega} \partial _{r} {\omega}_{2}\Big) \drm r \Bigg) \\
\Rightarrow  \partial _{z}{\chi}_{2} & = {\epsilon}\, \frac{FG}{H} \left(G^{2} \partial _{r} {\Omega}-F^{2}{\omega}_{2} \partial _{r} {\omega}_{1}-F^{2}{\omega}_{2}{\Omega} \partial _{r} {\omega}_{2}\right),
\end{align*}
with ${\epsilon}\in \{\pm 1\}$ only depending on the chosen orientation of our orthonormal tetrad. To proceed we need to specify our metric functions in order to calculate the twist potentials. 
First we are going to look at the asymptotics since they will give us important information later. 

\subsection{Asymptotic Minkowski Space-Times} \label{subsec:exasympt}
For an asymptotically-flat stationary and axisymmetric space-time in five dimensions we learn from \cite[Sec. IV.C]{Harmark:2004rm} the leading terms in the approach to Minkowski space. In ${\sigma}$-model 
form $\big($for $\sqrt{r^{2}+z^{2}}\to \infty $ and $z/\sqrt{r^{2}+z^{2}}$ finite$\big)$ the metric coefficients behave as follows 
\begin{equation} \label{eq:asymptmetr}
\begin{split}
J_{00} & = -1 + \frac{4M}{3{\pi}} \frac{1}{\sqrt{r^{2}+z^{2}}}+\mathcal{O}\left((r^{2}+z^{2})^{-1}\right),\\
J_{01} & = - \frac{L_{1}}{{\pi}} \frac{\sqrt{r^{2}+z^{2}}-z}{r^{2}+z^{2}}+\mathcal{O}\left((r^{2}+z^{2})^{-1}\right),\\
J_{02} & = - \frac{L_{2}}{{\pi}} \frac{\sqrt{r^{2}+z^{2}}+z}{r^{2}+z^{2}}+\mathcal{O}\left((r^{2}+z^{2})^{-1}\right), \\ 
J_{11} & =  \left(\sqrt{r^{2}+z^{2}}-z\right)\left[1+\frac{2}{3{\pi}}\frac{M+{\eta}}{\sqrt{r^{2}+z^{2}}}+\mathcal{O}\left((r^{2}+z^{2})^{-1}\right)\right],\\
J_{12} & = {\zeta} \frac{r^{2}}{\left(r^{2}+z^{2}\right)^{\frac{3}{2}}}+\mathcal{O}\left((r^{2}+z^{2})^{-1}\right),\\
J_{22} & = \left(\sqrt{r^{2}+z^{2}}+z\right)\left[1+\frac{2}{3{\pi}}\frac{M-{\eta}}{\sqrt{r^{2}+z^{2}}}+\mathcal{O}\left((r^{2}+z^{2})^{-1}\right)\right],\\
{\erm}^{2\nu} & = \frac{1}{2\sqrt{r^{2}+z^{2}}}+\mathcal{O}\left((r^{2}+z^{2})^{-1}\right).
\end{split}
\end{equation}
Here $M$ is the mass of the space-time and $L_{1}$, $L_{2}$ are the angular momenta; ${\zeta}$ and ${\eta}$ are constant where ${\eta}$ is not gauge-invariant, that is it changes under $z\to z+\text{const.}$, unlike ${\zeta}$; the periodicity of ${\varphi}$ and ${\psi}$ is assumed to be $2{\pi}$ (the case when it is $2{\pi}{\varepsilon}$ is given in \cite[Sec. IV.C]{Harmark:2004rm} as well).

Calculating the twist potentials on the top end rod by the method described above, we obtain to leading order in $z$ the expressions
\begin{equation*}
\left.{\chi}_{1}\right|_{r=0} \sim \frac{2{\epsilon}L_{1}}{πz}, \quad \left.{\chi}_{2}\right|_{r=0} \sim -\frac{4{\epsilon}ζ}{z},
\end{equation*}
where additive constants are dropped, so that both potentials go to zero at large $z$. Thus the patching matrix to leading order in $z$ beyond \eqref{eq:PMink} is
\begin{equation} \label{eq:asymptPtop}
\renewcommand\arraystretch{2.5}
P_{+} = \left(\begin{array}{ccc}-\dfrac{1}{2z} - \dfrac{M+{\eta}}{3{\pi}z^{2}} & \dfrac{{\epsilon}L_{1}}{{\pi}z^{2}} & -\dfrac{2{\epsilon}{\zeta}}{z^{2}} \\\dfrac{{\epsilon}L_{1}}{{\pi}z^{2}} & -1+ \dfrac{4M}{3{\pi}z} & -\dfrac{2L_{2}}{{\pi}z} \\-\dfrac{2{\epsilon}{\zeta}}{z^{2}} & -\dfrac{2L_{2}}{{\pi}z} & 2z + \dfrac{4(M-{\eta})}{3{\pi}}\end{array}\right).
\end{equation}
The subscript $+$ indicates that the patching matrix is adapted to the top asymptotic end. The adaptation $P_{-}$ to the bottom asymptotic end, that is the one which extends to $z\to -\infty $, is obtained by swapping ${\varphi}$ and ${\psi}$ in their roles. This leads to $z\mapsto -z$, $L_{1}\leftrightarrow L_{2}$. Furthermore, one has to check what happens with ${\zeta}$ and ${\eta}$ in this case. From \cite[Eq.~(5.18)]{Harmark:2004rm} we see that ${\zeta}\mapsto {\zeta}$ and ${\eta}\mapsto -{\eta}$ for the Myers-Perry solution. But all asymptotically flat space-times have the same fall off up to the order \eqref{eq:asymptmetr}, so this behaviour must be generic. For the ease of reference later on we will include $P_{-}$ again explicitly
\begin{equation}\label{eq:asymptPbot}
\renewcommand\arraystretch{2.5}
P_{-} = \left(\begin{array}{ccc}\dfrac{1}{2z} - \dfrac{M-{\eta}}{3{\pi}z^{2}} & -\dfrac{{\epsilon}L_{2}}{{\pi}z^{2}} & -\dfrac{2{\epsilon}{\zeta}}{z^{2}} \\-\dfrac{{\epsilon}L_{2}}{{\pi}z^{2}} & -1- \dfrac{4M}{3{\pi}z} & \dfrac{2L_{1}}{{\pi}z} \\-\dfrac{2{\epsilon}{\zeta}}{z^{2}} & \dfrac{2L_{1}}{{\pi}z} & -2z + \dfrac{4(M+{\eta})}{3{\pi}}\end{array}\right).
\end{equation}

The Myers-Perry solution \cite{Myers:1986aa}, which we will study next, is the five-dimensional pendant of the Kerr solutions, that is it describes a five-dimensional spinning black hole with a topologically spherical horizon. 

\subsection{Five-Dimensional Myers-Perry Solution} \label{subsec:PMP}
The calculation in the first part of this example up to the expression for $J(r,z)$ is based on \cite{Harmark:2004rm}. The Myers-Perry metric is given by
\begin{equation} \label{eq:MPmetric}
\begin{split}
\drm s^{2} & = -\drm t^{2} + \frac{{\rho}_{0}^{2}}{{\Sigma}}\left[\drm t - a_{1} \sin^{2}{\theta} \,\drm {\varphi} - a_{2} \cos^{2} {\theta} \,\drm {\psi}\right]^{2} \\
& \hspace{0.4cm} + ({\rho}^{2}+a_{1}^{2})\sin^{2}{\theta} \,\drm {\varphi}^{2} + ({\rho}^{2}+a_{2}^{2})\cos^{2}{\theta} \,\drm {\psi}^{2} \\
& \hspace{0.4cm} + \frac{{\Sigma}}{{\Delta}} \,\drm{\rho}^{2} + {\Sigma} \,\drm {\theta}^{2},
\end{split}
\end{equation}
where
\begin{equation} \label{eq:MPnot}
\begin{split}
{\Delta} & = {\rho}^{2} \left(1+\frac{a_{1}^{2}}{{\rho}^{2}}\right)\left(1+\frac{a_{2}^{2}}{{\rho}^{2}}\right) - {\rho}_{0}^{2}, \\
{\Sigma} & = {\rho}^{2} + a_{1}^{2} \cos^{2}{\theta} + a_{2}^{2} \sin^{2}{\theta},
\end{split}
\end{equation}
and the coordinate ranges are
\begin{equation*}
t\in \mathbb{R},\quad {\varphi},{\psi} \in  \left[0,2{\pi}\right), \quad {\theta} \in  [0,{\pi}].
\end{equation*}
The Weyl coordinates can be taken to be
\begin{equation*}
r = \frac{1}{2} {\rho} \sqrt{{\Delta}} \sin 2{\theta}, \quad z = \frac{1}{2} {\rho}^{2} \left(1+\frac{a_{1}^{2}+a_{2}^{2}-{\rho}_{0}^{2}}{2{\rho}^{2}}\right) \cos 2{\theta}.
\end{equation*}
The rod structure consists of three components $(-\infty ,{\alpha})$, $(-{\alpha},{\alpha})$, $({\alpha},\infty )$, where
\begin{equation*}
{\alpha} = \frac{1}{4}\sqrt{\left({\rho}_{0}^{2}-a_{1}^{2}-a_{2}^{2}\right)^{2}-4a_{1}^{2}a_{2}^{2}}.
\end{equation*}
The rod vectors turn out to be as follows.
\begin{enumerate}[(1)]
\item If $z$ lies in the semi-infinite spacelike rod $({\alpha},\infty )$, the rod vector is $∂_{\varphi}$.
\item If $z$ lies in the finite timelike rod $(-{\alpha},{\alpha})$, the kernel of $J$ is spanned by the vector
\begin{equation*}
\left(\begin{array}{ccc}1 & {\Gamma}_{1} & {\Gamma}_{2}\end{array}\right)^{\mathrm{t}},
\end{equation*}
in the basis $(\partial_t,\partial_\varphi,\partial_\psi)$, where ${\Gamma}_{1,2}$ are the angular velocities
\begin{equation*}
{\Gamma}_{1}= \frac{{\rho}_{0}^{2}+a_{1}^{2}-a_{2}^{2}-4{\alpha}}{2a_{1}^{\vphantom{1}}{\rho}_{0}^{2}}, \quad {\Gamma}_{2}= \frac{{\rho}_{0}^{2}-a_{1}^{2}+a_{2}^{2}-4{\alpha}}{2a_{2}^{\vphantom{1}}{\rho}_{0}^{2}}.
\end{equation*}
This rod corresponds to an event horizon with topology $S^{3}$ (see \cite{Hollands:2008fp}, proof of Proposition 2 in Section 3).
\item If $z$ lies in the semi-infinite spacelike rod $(-\infty ,-{\alpha})$, the rod vector is $∂_{\psi}$.
\end{enumerate}
The conserved Komar quantities are
\begin{equation} \label{eq:MPmassmom}
M = \frac{3{\pi}}{8} {\rho}_{0}^{2}, \quad L_{1}^{\vphantom{1}} = \frac{3{\pi}}{8} a_{1}^{\vphantom{1}} {\rho}_{0}^{2}, \quad L_{2}^{\vphantom{1}} = \frac{3{\pi}}{8} a_{2}^{\vphantom{1}} {\rho}_{0}^{2}.
\end{equation}

Now we can again calculate the twist potentials on the top end rod as shown earlier and obtain
\begin{equation*}
\left.{\chi}_{1}\right|_{{\theta}=0} = -\frac{{\epsilon}{\rho}_{0}^{2}a_{1}^{\vphantom{1}}}{{\rho}^{2}+a_{1}^{2}}, \quad \left.{\chi}_{2}\right|_{{\theta}=0} = \frac{{\epsilon}a_{1}^{\vphantom{1}}a_{2}^{\vphantom{1}}{\rho}_{0}^{2}}{{\rho}^{2}+a_{1}^{2}}.
\end{equation*}
On ${\theta}=0$ we have
\begin{equation*}
{\rho}^{2} = 2z + \frac{1}{2} \left({\rho}_{0}^{2}-a_{1}^{2}-a_{2}^{2}\right),
\end{equation*}
so the notation in the calculation of $P$ can be somewhat streamlined by introducing
\begin{equation*}
{\beta} = \frac{1}{4} \left(-{\rho}_{0}^{2}+a_{1}^{2}-a_{2}^{2}\right),\ {\gamma} = \frac{1}{4} \left({\rho}_{0}^{2}+a_{1}^{2}-a_{2}^{2}\right).
\end{equation*}
Then a straightforward computation shows
\begin{equation} \label{eq:MPpatmat} 
\renewcommand\arraystretch{2.5}
P_{1} = \left(\begin{array}{ccc}-\dfrac{z+{\gamma}}{2(z^{2}-{\alpha}^{2})} & -\dfrac{{\rho}_{0}^{2}a_{1}^{\vphantom{1}}}{4(z^{2}-{\alpha}^{2})} & \dfrac{{\rho}_{0}^{2}a_{1}^{\vphantom{1}}a_{2}^{\vphantom{1}}}{4(z^{2}-{\alpha}^{2})} \\
\cdot  & -\dfrac{z^{2}+z({\beta}-{\gamma})+{\gamma}^{2}-{\beta}{\gamma}-{\alpha}^{2}}{z^{2}-{\alpha}^{2}} & -\dfrac{a_{2}^{\vphantom{1}}{\rho}_{0}^{2}(z-{\gamma})}{2(z^{2}-{\alpha}^{2})} \\
\cdot  & \cdot  & \hspace{-1.2cm} 2(z-{\beta})+\dfrac{a_{2}^{2}{\rho}_{0}^{2}(z-{\gamma})}{2(z^{2}-{\alpha}^{2})}\end{array}\right)_{\vphantom{\frac{1}{2}}},
\end{equation}
where the subscript indicates that it is adapted to rod 1 according to the numbering above.

In the case of $a_{1}=a_{2}=0$ the Myers-Perry metric becomes the 5-dimensional Schwarzschild metric
\begin{align*}
\drm s^{2} & = \left(-1+\frac{{\rho}_{0}^{2}}{{\rho}^{2}}\right)\drm t^{2} + {\rho}^{2}\sin^{2}{\theta} \,\drm {\varphi}^{2} + {\rho}^{2}\cos^{2}{\theta} \,\drm {\psi}^{2} \\
& \hspace{0.4cm} + \left(1-\frac{{\rho}_{0}^{2}}{{\rho}^{2}}\right)^{-1} \,\drm {\rho}^{2} + {\rho}^{2} \,\drm {\theta}^{2},
\end{align*}
The twist potentials are globally constant and we set them without loss of generality to zero. The adaptations to the three different parts of the axis, then take the following form.
\begin{enumerate}[(1)]
\item Spacelike rod $z\in ({\alpha},\infty )$: 
\begin{equation*}
P_{1}(z)=\mathrm{diag}\,\left(-\frac{1}{2(z-\alpha)}, -\frac{z-\alpha}{z+\alpha},2(z+{\alpha})\right).
\end{equation*}
\item Horizon rod $z\in (-{\alpha},{\alpha})$: 
\begin{equation*}
P_{2}(z)=\mathrm{diag}\,\left(-\frac{1}{4(z^{2}-\alpha^{2})}, -2(z-{\alpha}),2(z+{\alpha})\right).
\end{equation*}
\item Spacelike rod $z\in (-\infty ,-{\alpha})$: 
\begin{equation*}
P_{3}(z)=\mathrm{diag}\,\left(\frac{1}{2({z+\alpha})}, -\frac{z+\alpha}{z-{\alpha}},-2(z-{\alpha})\right).
\end{equation*}
\end{enumerate}

\subsection{Black Ring Solutions}
The five-dimensional black ring of Emparan and Reall \cite{Emparan:2002aa} is a space-time with a black hole whose horizon has topology $S^{1}\times S^{2}$. 
We shall take formulae and notation from \cite[Sec.~VI]{Harmark:2004rm}.


The metric is
\begin{equation} \label{eq:BRmetric} 
\begin{split}
\drm s^{2} & = -\frac{F(v)}{F(u)} \left(\drm t - C {\kappa} \frac{1+v}{F(v)}\,\drm {\varphi}\right)^{2} \\
		& \hspace{0.4cm} +\frac{2{\kappa}^{2}F(u)}{(u-v)^{2}}\left[-\frac{G(v)}{F(v)}\,\drm {\varphi}^{2} + \frac{G(u)}{F(u)}\,\drm {\psi}^{2} + \frac{1}{G(u)} \,\drm u^{2}-\frac{1}{G(v)} \,\drm v^{2} \right],
\end{split}
\end{equation}
where $F({\xi})$ and $G({\xi})$ are 
\begin{equation*}
F({\xi})=1+b{\xi},\quad G({\xi})= (1-{\xi}^{2})(1+c{\xi}),
\end{equation*}
and the parameters vary in the ranges
\begin{equation*}
0<c\leq b<1.
\end{equation*}
The parameter ${\kappa}$ has the dimension of length and for thin rings it is roughly the radius of the ring circle. The constant $C$ is given in terms of $b$ and $c$ by 
\begin{equation*}
C=\sqrt{2b(b-c)\frac{1+b}{1-b}},
\end{equation*}
and the coordinate ranges for $u$ and $v$ are
\begin{equation*}
-1\leq u\leq 1, \quad -\infty \leq v\leq -1
\end{equation*}
with asymptotic infinity recovered as $u\to v\to -1$. For the ${\varphi}$-coordinate the axis of rotation is $v=-1$, and for the ${\psi}$-direction the axis is divided in two components. First $u=1$ which is the disc bounded by the ring, and second $u=-1$ which is the outside of the ring, that is up to infinity. The horizon is located at $v=-\frac{1}{c}$ and outside of it at $v=-\frac{1}{b}$ lies an ergosurface. As argued in \cite[Sec.~5.1.1]{Emparan:2008aa} three independent parameters $b$, $c$, ${\kappa}$ is one too many, since for a ring with a certain mass and angular momentum we expect its radius to be dynamically fixed by the balance between centrifugal and tensional forces. This is here the case as well, because in general there are conical singularities on the plane containing the ring, $u=\pm 1$. In order to cure them ${\varphi}$ and ${\psi}$ have to be identified with periodicity 
\begin{equation*}
{\Delta}{\varphi}={\Delta}{\psi}=4{\pi} \frac{\sqrt{F(-1)}}{|G'(-1)|}=2{\pi}\frac{\sqrt{1-b}}{1-c},
\end{equation*}
and the two parameters have to satisfy
\begin{equation} \label{eq:BRb}
b=\frac{2c}{1+c^{2}}. 
\end{equation}
This leaves effectively a two-parameter family of solutions as expected with the Killing vector fields $X_{0}=∂_{t}$, $X_{1}=∂_{\varphi}$ and $X_{2}=∂_{\psi}$. 
For the moment, however, we will keep the conical singularity in and regard the parameter $b$ as free. It can be eliminated at any time using \eqref{eq:BRb}.

A straightforward calculation shows
\begin{equation*}
\det J=\frac{4{\kappa}^{4}}{(u-v)^{4}}G(u)G(v), 
\end{equation*}
hence we define
\begin{equation*}
r=\frac{2{\kappa}^{2}}{(u-v)^{2}}\sqrt{-G(u)G(v)}. 
\end{equation*}
The harmonic conjugate can be calculated in the same way as for the Myers-Perry solution (for details see \cite[App.~H]{Harmark:2004rm}) and one obtains
\begin{equation*}
z=\frac{{\kappa}^{2}(1-uv)(2+cu+cv)}{(u-v)^{2}}.
\end{equation*}
Using expressions for $u$, $v$ in terms of $r$, $z$ (see \cite[App.~H]{Harmark:2004rm})
\begin{align*}
u & = \frac{(1-c)R_{1}-(1+c)R_{2}-2R_{3}+2(1-c^{2}){\kappa}^{2}}{(1-c)R_{1}+(1+c)R_{2}+2cR_{3}} \\
v & = \frac{(1-c)R_{1}-(1+c)R_{2}-2R_{3}-2(1-c^{2}){\kappa}^{2}}{(1-c)R_{1}+(1+c)R_{2}+2cR_{3}},
\end{align*}
where
\begin{equation*} 
R_{1}=\sqrt{r^{2}+(z+c{\kappa}^{2})^{2}},\ R_{2}=\sqrt{r^{2}+(z-c{\kappa}^{2})^{2}},\ R_{3}=\sqrt{r^{2}+(z-{\kappa}^{2})^{2}},
\end{equation*}
the $J$-matrix can be computed as 
\begin{align*} 
J_{00} & = -\frac{(1+b)(1-c)R_{1}+(1-b)(1+c)R_{2}-2(b-c)R_{3}-2b(1-c^{2}){\kappa}^{2}}{(1+b)(1-c)R_{1}+(1-b)(1+c)R_{2}-2(b-c)R_{3}+2b(1-c^{2}){\kappa}^{2}}, \\
 J_{01} & = -\frac{2C{\kappa}(1-c)[R_{3}-R_{1}+(1+c){\kappa}^{2}}{(1+b)(1-c)R_{1}+(1-b)(1+c)R_{2}-2(b-c)R_{3}+2b(1-c^{2}){\kappa}^{2}}, \\
 J_{22} & = \frac{(R_{3}+z-{\kappa}^{2})(R_{2}-z+c{\kappa}^{2})}{R_{1}-z-c{\kappa}^{2}}, \\
 J_{11} & = -\frac{r^{2}}{J_{00}J_{22}}+\frac{J_{01}^{2}}{J_{00}},
\end{align*}
with the remaining components vanishing, and
\begin{align*} 
{\erm}^{2ν} & = \left[(1+b)(1-c)R_{1}+(1-b)(1+c)R_{2}+2(c-b)R_{3}+2b(1-c^{2}){\kappa}^{2}\right]\\
 & \hspace{0.4cm} ×\frac{(1-c)R_{1}+(1+c)R_{2}+2cR_{3}}{8(1-c^{2})^{2}R_{1}R_{2}R_{3}}.
\end{align*}

The rod structure consists of four components $(-\infty ,-c{\kappa}^{2})$, $(-c{\kappa}^{2},c{\kappa}^{2})$, $(c{\kappa}^{2},{\kappa}^{2})$, $({\kappa}^{2},\infty )$. 
\begin{enumerate}[(1)]
\item For $r=0$ and $z\in ({\kappa}^{2},\infty )$ we have $R_{3}-R_{1}+(1+c){\kappa}^{2}=0$ which implies $J_{01}=J_{11}=0$. Hence, the interval $({\kappa}^{2},\infty )$ is a semi-infinite spacelike rod in direction $∂_{\varphi}$.
\item For $r=0$ and $z\in (c{\kappa}^{2},{\kappa}^{2})$ we have $R_{2}+R_{3}-(1-c){\kappa}^{2}=0$ which implies $J_{22}=0$. Hence, the interval $(c{\kappa}^{2},{\kappa}^{2})$ is a finite spacelike rod in direction $∂_{\psi}$.
\item For $r=0$ and $z\in (-c{\kappa}^{2},c{\kappa}^{2})$ we have $R_{1}+R_{2}-2c{\kappa}^{2}=0$ which implies that the kernel of $J$ in this range is spanned by the vector
\begin{equation*}
\left(\begin{array}{ccc}1 & {\Gamma} & 0\end{array}\right)^{\mathrm{t}}, \quad \text{ where } {\Gamma}=\frac{b-c}{(1-c)Cκ}
\end{equation*}
is again the angular velocity. Thus, $(-c{\kappa}^{2},c{\kappa}^{2})$ is a finite timelike rod and it can be shown that it corresponds to an event horizon with topology $S^{2}\times S^{1}$ (a brief reasoning can be found in \cite{Hollands:2008fp}, proof of Proposition~2 in Section~3).
\item For $r=0$ and $z\in (-\infty ,-c{\kappa}^{2})$ we have $R_{1}-R_{3}+(1+c){\kappa}^{2}=0$ which implies $J_{22}=0$. Hence, the interval $(-\infty ,-c{\kappa}^{2})$ is a semi-infinite spacelike rod in direction $∂_{\psi}$.
\end{enumerate}

As before we compute the patching matrix. However, this time some of the metric components vanish and the metric can be written as
\begin{align*}
\drm s^{2} & = J_{00} \, \drm t^{2} +2 J_{01} \, \drm t \drm {\varphi} + J_{11} \, \drm {\varphi}^{2} + J_{22} \, \drm {\psi}^{2} + \erm^{2{\nu}}\left(\drm r^{2}+\drm z^{2}\right)\\
& = -F^{2} \left(\drm t + {\omega} \, \drm {\varphi} \right)^{2} + G^{2} \, \drm {\varphi}^{2}+ H^{2} \, \drm {\psi}^{2} + \erm^{2{\nu}}\left(\drm r^{2}+\drm z^{2}\right)
\end{align*}
with
\begin{equation*}
F^{2} = - J_{00},\ F^{2}{\omega} = -J_{01},\ G^{2}-F^{2}{\omega}^{2} = J_{11},\ H^{2} = J_{22}.
\end{equation*}
We see immediately that only one of the twist 1-forms is non-vanishing. On the top end rod $∂_{\varphi}=0$ so that the relevant Killing 1-forms are
\begin{align*}
\frac{\partial }{\partial t} & \to  T = -F^{2} \left(\drm t + {\omega} \, \drm {\varphi} \right) = - F {\theta}^{0} \\
\frac{\partial }{\partial {\psi}} & \to  Ψ = H^{2} \, \drm {\psi} = H {\theta}^{2},
\end{align*}
where we used again the orthonormal basis
\begin{equation*} 
{\theta}^{0}= F\left(\drm t + {\omega} \, \drm {\varphi} \right),\ {\theta}^{1} = G\, \drm {\varphi},\ {\theta}^{2}= H\, \drm {\psi},\ {\theta}^{3}= \erm^{{\nu}} \, \drm r,\  {\theta}^{4}= \erm^{{\nu}} \, \drm z.
\end{equation*}
For the twist 1-form we then get
\begin{align*}
\drm {\chi} & =  \partial _{r}{\chi}\, \drm r + \partial _{z} {\chi}\, \drm z = * (T\wedge Ψ\wedge \drm T)\\
& = * (F {\theta}^{0}\wedge H{\theta}^{2}\wedge F^{2}\, \drm {\omega} \wedge G^{-1}{\theta}^{1})\\
& = - \frac{F^{3}H}{G} *({\theta}^{0}\wedge {\theta}^{1}\wedge {\theta}^{2}\wedge \drm {\omega}) \\
& = - \frac{F^{3}H}{G} *\Big({\theta}^{0}\wedge {\theta}^{1}\wedge {\theta}^{2}\wedge (\partial _{r}{\omega}\, \drm r+\partial _{z}{\omega}\, \drm z)\Big),
\end{align*}
thus
\begin{equation*}
\partial _{r} {\chi} = {\epsilon} \frac{F^{3}H}{G} \partial _{z} {\omega},\quad \partial _{z} {\chi} = - {\epsilon} \frac{F^{3}H}{G} \partial _{r} {\omega}. 
\end{equation*}
Note that
\begin{equation*}
{\omega} = \frac{J_{01}}{J_{00}}, \quad G^{2} = - \frac{r^{2}}{J_{00}J_{22}},
\end{equation*}
as $-r^{2} = \det J = (J_{00}^{\vphantom{1}}J_{11}^{\vphantom{1}}-J_{01}^{2})J_{22}^{\vphantom{1}}$. On $r=0$ we also see that
\begin{equation*}
R_{1} = |z+c{\kappa}^{2}|, \quad R_{2} = |z-c{\kappa}^{2}|, \quad R_{3} = |z-{\kappa}^{2}|,
\end{equation*}
and for ${\kappa}^{2}<z<\infty $ the moduli signs can be dropped. Then the metric coefficients behave as 
\begin{equation*}
J_{00} = \mathcal{O}(1),\quad J_{01} = \mathcal{O}(r^{2}),\quad J_{22}=\mathcal{O}(1), \quad {\omega}^{2}= \mathcal{O}(r^{2}),
\end{equation*}
so that we obtain
\begin{align*}
\partial _{z} {\chi} & = - {\epsilon} \frac{(-J_{00})^{\frac{3}{2}}(J_{22})^{\frac{1}{2}}}{r} (-J_{00})^{\frac{1}{2}} (J_{22})^{\frac{1}{2}}\, \partial _{r} \left(\frac{J_{01}}{J_{00}}\right) \\
& = -{\epsilon} \frac{J_{00}^{2}J_{22}^{\vphantom{1}}}{r}\, \partial _{r}\left(\frac{J_{01}^{\vphantom{1}}}{J_{00}^{\vphantom{1}}}\right).
\end{align*}
Now, if $J_{01}=r^{2}B(z)+\mathcal{O}(r^{4})$, then
\begin{equation} \label{eq:BRtwistpot1}
\lim_{r\to 0} \partial _{z}{\chi} = -{\epsilon} \lim_{r\to 0} 2J_{00}J_{22}B(z).
\end{equation}
In order to determine $B(z)$ we do some auxiliary calculations. Denote $α=cκ^{2}$, $β=κ^{2}$. Then with $z>β$ and to leading order in $r$ it is
\begin{align*}
R_{1} & = (z+α)\left(1+\frac{r^{2}}{2(z+α)^{2}}\right),\ R_{2} = (z-α)\left(1+\frac{r^{2}}{2(z-α)^{2}}\right) \\
R_{3} & = (z-β)\left(1+\frac{r^{2}}{2(z-β)^{2}}\right),
\end{align*}
whence
\begin{equation*}
J_{22} = 2 (z-β)\frac{2(z+α)}{r^{2}}\frac{r^{2}}{2(z-α)} = \frac{2(z-β)(z+α)}{z-α}. 
\end{equation*}
Second we compute
\begin{equation} \label{eq:BRlambda}
J_{00} = -\frac{z-α}{z+λ}, \quad \text{where } λ= κ^{2}⋅\frac{2b-bc-c}{1-c}.
\end{equation}
Last, we obtain
\begin{equation*}
J_{01} = -\frac{C(1-c)\kappa^3}{2(1-b)}\frac{1}{(z-β)(z+α)(z+\lambda)}⋅r^{2}. 
\end{equation*}
Using these results \eqref{eq:BRtwistpot1} can be integrated to
\begin{equation*}
\left.{\chi}\right|_{r=0} = \frac{2\nu}{z+\lambda},\quad \nu=\frac{\epsilon C(1-c)\kappa^3}{1-b}.
\end{equation*}
Note that this agrees up to a constant with \cite[Eq.~(25)]{Tomizawa:2004aa}. Now we can compute the quantities which go in the patching matrix. 
The restriction $r=0$ is not explicitly mentioned, but still assumed in the following.
\begin{align*}
g{\chi} & = \frac{{\chi}}{J_{00}J_{22}} = -\frac{ν}{(z-β)(z+α)},\\
g & = \frac{1}{J_{00}J_{22}} = -\frac{z+λ}{2(z+α)(z-β)},
\end{align*}
For the last matrix entry we first calculate some auxiliary quantities. From \eqref{eq:BRlambda} we obtain
\begin{equation*}
b = \frac{λ+α}{λ+2β-α}, 
\end{equation*}
hence
\begin{equation*} 
b-c=\frac{(β-α)(λ-α)}{β(λ+2β-α)}, \quad 1+b = \frac{2(λ+β)}{λ+2β-α}, \quad 1-b = \frac{2(β-α)}{λ+2β-α}.
\end{equation*}
This yields
\begin{equation*}
2ν^{2} = \frac{4b(b-c)(1+b)(1-c)^{2}κ^{6}}{(1-b)^{3}} = (λ+α)(λ-α)(λ+β),
\end{equation*}
which in turn justifies the following factorization
\begin{equation*} 
(z-α)(z+α)(z-β)+2ν^{2} = \big(z+λ\big)\big(z^2-(β+λ)z -α^2+β λ+λ^2\big).
\end{equation*}
and eventually
\begin{align*}
J_{00}+g{\chi}^{2} & = -\frac{z-α}{z+λ} - \frac{2ν^{2}}{(z+λ)(z+α)(z-β)} \\
& = -\frac{z^2-(β+λ)z -α^2+β λ+λ^2}{(z+α)(z-β)}.
\end{align*}
The patching matrix for $z\in (β,\infty )$ and $r=0$ is now 
\begin{equation} \label{eq:BRpatmat}
\renewcommand\arraystretch{2.5}
P_{1}=\left(\begin{array}{ccc}-\dfrac{z+{\lambda}}{2(z+α)(z-β)} & \dfrac{{\nu}}{(z+α)(z-β)} & 0 \\
\cdot  & -\dfrac{z^{2}-{\gamma}z+{\delta}}{(z+α)(z-β)} & 0 \\
0 & 0 & \dfrac{2(z+α)(z-β)}{z-α}\end{array}\right),
\end{equation}
where the index again only indicates that it is adapted to the part of the axis which extends to $+\infty $ and where
\begin{equation} \label{eq:BRpar}
\renewcommand\arraystretch{2.5}
\begin{array}{ccc} α=c{\kappa}^2, & β={\kappa}^2, & {\lambda}=κ^{2}⋅\dfrac{2b-bc-c}{1-b}, \\
{\nu}=\dfrac{\epsilon C(1-c)\kappa^3}{1-b}, & {\gamma} = κ^{2}+λ, & {\delta}=-c^{2}κ^{4}+κ^{2} λ+λ^2 \end{array}.
\end{equation}
Note that this is based on the assumption that the periodicity of ${\varphi}$, ${\psi}$ is $2{\pi}$, otherwise it has to be modified according to \cite[Eq.~(4.17)]{Harmark:2004rm}.

From \eqref{eq:asymptPbot} we read off the conserved Komar quantities as
\begin{equation*} 
M = \frac{3{\pi}}{4}(λ+c^{2}κ^{4}), \quad L_{1} = \dfrac{π C(1-c)\kappa^3}{1-b}, \quad L_{2}=0.
\end{equation*}

%% file: TheConverse.tex
\section{The Converse} \label{sec:converse}
As already mentioned in the introduction, the following is an immediate consequence of the twistor construction that we described in Part~I. 
\begin{cor}
The patching matrix $P$ (adapted to any portion of the axis $r=0$) determines the metric and conversely. 
\end{cor}
\begin{proof}[Sketch of Proof]
$J'(r,z)$ is obtained from $P(w)$ by the splitting procedure, see Section~3 in Part~I, and conversely $P(w)$ is the analytic continuation of $J'(r=0,z)$.
\end{proof}

It is known that the classification of black holes in four dimensions does not straight-forwardly generalize to five dimensions. The Myers-Perry solution and the black ring are 
space-times whose range of parameters (mass and angular momenta) do have a non-empty intersection, but their horizon topology is different, which means they cannot be isometric. 
In order to address this issue the rod structure is introduced to supplement the set of parameters. 

Using this extended set of parameters the following theorem from \cite{Hollands:2008fp} is a first step towards a classification.

\begin{thm} \label{thm:holluniqueness}
Two five-dimensional, asymptotically flat vacuum space-times with connected horizon where each of the space-times admits three commuting Killing vector fields, one time translation 
and two axial Killing vector fields, are isometric if they have the same mass and two angular momenta, and their rod structures coincide.
\end{thm}

Note, however, that \cite[Prop.~3.1]{Chrusciel:2011eu} suggests that by adding the rod structure to the list of parameters the mass becomes redundant, at least for connected 
horizon.

Theorem~\ref{thm:holluniqueness} answers the question about uniqueness of five-dimensional black holes, but not existence. In other words, we do not yet know which combinations of 
rod structure and angular momenta are permitted, and how they determine the twistor data, that is essentially $P$, and thereby the metric. It is natural to conjecture:

\begin{conj}
Rod structure and angular momenta determine $P$ (even for a disconnected horizon). 
\end{conj}

\subsection[From Rod Structure to Patching Matrix]{From Rod Structure to Patching Matrix --- an Ansatz}

In the following we will present an ansatz for this reconstruction of the patching matrix from the given data, exemplified in cases where the rod structure has up to three 
nuts. 

Given a rod structure with nuts at $\{a_{i} | a_{i}∈ℝ\}_{1≤i≤N}$ we know that $P$ can at most have single poles at these nuts, see Corollary~6.3 and Proposition~6.4 in Part~I.
We shall see that this fact can also be derived from the switching procedure (Theorem~\ref{thm:switching}) and thus we make the ansatz
\begin{equation*}
P(z) = \frac{1}{{\Delta}} P'(z), 
\end{equation*}
where ${\Delta}=\prod _{i=1}^{N}(z-a_{i})$ and the entries of $P'(z)$ are polynomials in $z$. If we now choose $P$ to be adapted to the top outermost rod 
$(a_{\scriptscriptstyle N},\infty )$, then Section~\ref{subsec:exasympt} tells us its asymptotic behaviour as $z → ∞$, that is $P$ asymptotes $P_{+}$ given in \eqref{eq:asymptPtop}. 
This implies that the entries of $P'(z)$ are in fact polynomials of the following degrees,
\begin{equation*}
\renewcommand{\arraystretch}{1.5}
P'(z) = \left(\begin{array}{ccc} q_{\scriptscriptstyle N-1}(z) & q_{\scriptscriptstyle N-2}(z) & q_{\scriptscriptstyle N-2}(z) \\ \cdot  & q_{\scriptscriptstyle N}(z) & 
q_{\scriptscriptstyle N-1}(z) \\ \cdot  & \cdot  & q_{\scriptscriptstyle N+1}(z)\end{array}\right),
\end{equation*}
where $q_{k}$ is a polynomial of degree $k$. (Here the notation shall just indicate the degree of the polynomials, that is two appearances of $q_{\scriptscriptstyle N-1}$ or 
$q_{\scriptscriptstyle N-2}$ in different entries of the matrix can still be different polynomials, and if $N-2<0$ then it shall be the zero-polynomial.) In fact, 
from \eqref{eq:asymptPtop} we can not only deduce the degree of the polynomials but also their leading coefficients. The diagonal entries will have leading coefficient 
$-\frac{1}{2}$, $-1$, and 2, respectively, and the leading coefficients on the superdiagonal will be proportional to the angular momenta. Similarly, one can use 
\eqref{eq:asymptPbot} for $P$ adapted to the bottom outermost rod $(-∞, a_{1})$. Note that this does not impose any further restrictions on the coefficients of the 
space-time metric apart from being analytic.

The number of free parameters in $P$ equals the number of independent coefficients in the polynomials. Our aim must be to tie down our space-time metric by fixing all those 
parameters in terms of the $a_{i}$ and the angular momenta $L_{1}$, $L_{2}$. Any free parameter left in $P$ is then a free parameter in our (family of) solutions. 

\begin{ex}[One-Nut Rod Structure]\mbox{}\\
Consider the case where the rod structure has one nut, which is without loss of generality at the origin (remember that a shifted rod structure corresponds to a 
diffeomorphic space-time), see Figure~\ref{fig:onenutrodstr}. We do not make assumptions about the angular momenta $L_{1}$, $L_{2}$. 

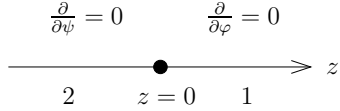
\begin{figure}[htbp]
\begin{center}
     \scalebox{0.8}{\input{figure1.tex}}
     \caption{Rod Structure with one nut at the origin. The numbers are just labelling the parts of the axis.} 
     \label{fig:onenutrodstr}
\end{center}
\end{figure}
According to our ansatz we have for the patching matrix on the top part of the axis
\begin{equation*}
\renewcommand{\arraystretch}{1.5}
P_{1} = \frac{1}{z} 
\left(\begin{array}{ccc}
-\dfrac{1}{2}  & 0 & 0 \\
0 & -z + c_{1} & c_{2} \\
0 & c_{2} & 2z^{2} + c_{3}z+ c_{4}
\end{array}\right),
\end{equation*}
which implies $L_{1}=ζ=0$. On the other hand for the bottom part it is
\begin{equation*}
\renewcommand{\arraystretch}{1.5}
P_{2} = \frac{1}{z} 
\left(\begin{array}{ccc}
-\dfrac{1}{2}  & 0 & 0 \\
0 & -z + \tilde c_{1} & \tilde c_{2} \\
0 & \tilde c_{2} & 2z^{2} + \tilde c_{3}z+ \tilde c_{4}
\end{array}\right),
\end{equation*}
and therefore necessarily $L_{2}=ζ=0$, thus $c_{2}=0$. This forces the patching matrix to be diagonal and since it has to have unit determinant,
\begin{equation*}
\det P_{1} = \frac{1}{z^{3}} \left(-\frac{1}{2} \right)\left(-z + c_{1}\right)\left(2z^{2} + c_{3}z+ c_{4}\right) = 1,
\end{equation*}
we obtain $c_{1}=c_{3}=c_{4}=0$. But this is the patching matrix for flat space, see \eqref{eq:PMink}. 

Hence we have shown that for a rod structure with one nut not all values for the conserved quantities are allowed, in fact they all (including mass) have to vanish, 
which in turn uniquely determines the space-time as Minkowski space.\\ \mbox{} \hfill $\blacksquare$
\end{ex}

Attempting the same for a rod structure with two nuts one will quickly notice that more tools are necessary in order to fix all the parameters. Here 
Corollary~6.7 in Part~I is useful.
\begin{cor} \label{cor:invpmatrix}
In five space-time dimensions, if $P_{+}$ is the patching matrix adapted to $(a_{\scriptscriptstyle N}, ∞)$, then 
${\Delta} \coloneqq\prod _{i=1}^{N}(z-a_{i})$ divides all $2×2$-minors of $Δ⋅P_{+}^{\vphantom{-1}}=P'_{+}$.
\end{cor}
In Part~I we have also seen that this guarantees the metric coefficients on $(a_{\scriptscriptstyle N}, ∞)$ to be bounded for 
$z \downarrow a_{\scriptscriptstyle N}$. However, despite the regularity of the metric, this does not have to hold for the other nuts, as we 
have seen for example for the black ring. 
\begin{ex}[Two-Nut Rod Structure] \mbox{}\\
Consider the rod structure as in Figure~\ref{fig:twonutrodstr}. 
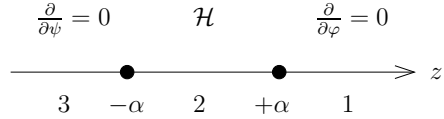
\begin{figure}[htbp]
\begin{center}
     \scalebox{0.8}{\input{figure2.tex}}
     \caption{Rod structure with two nuts.} 
     \label{fig:twonutrodstr}
\end{center}
\end{figure}

In line with the above ansatz we start off from
\begin{equation} \label{eq:MPpatmatansatz}
\renewcommand{\arraystretch}{2.5}
P=\frac{1}{z^{2}-α^{2}} \left(\begin{array}{ccc}-\dfrac{1}{2}z+c_{1} & \dfrac{L_{1}}{π} & c_{2}\\ \cdot  & -z^{2}+c_{3}z+c_{4} 
& -\dfrac{2L_{2}}{π}z+c_{5} \\ \cdot  & \cdot  & 2z^{3}+c_{6}z^{2}+c_{7}z+c_{8}\end{array}\right),
\end{equation}
which we assume to be adapted to the top section of the axis $(\alpha,∞)$ and where the orientation of the basis is without loss of generality chosen such that 
$\epsilon=1$ in \eqref{eq:asymptPtop}. Nondegeneracy requires $\alpha>0$.

One restriction on the constants is immediate from (\ref{eq:MPpatmatansatz}): the top-left entry must not change sign on the top rod so that
\begin{equation}\label{eq:MPcoeff0}
\alpha>2c_1.
\end{equation}
We now make use of Corollary~\ref{cor:invpmatrix} which for the minor obtained by cancelling the third row and first column yields
\begin{equation*}
c_{2}z^{2} + \left(-c_{2}c_{3} - \frac{2L_{1}L_{2}}{π^{2}}\right)z-c_{2}c_{4}+\frac{L_{1}c_{5}}{π} \sim z^{2}-α^{2}, 
\end{equation*}
where $\sim$ means that the left hand side has a factor $z^{2}-α^{2}$. Assume that $L_{1}, L_{2}≠0$, then comparing the (ratio of) 
coefficients gives
\begin{align}
c_{2}c_{3} & = -\frac{2L_{1}L_{2}}{π^{2}}, \label{eq:MPcoeff1}\\
c_{4} & = α^{2} +\frac{L_{1}c_{5}}{πc_{2}}.\label{eq:MPcoeff2}
\end{align}

Choosing the minor obtained from cancelling the second row and third column we get
\begin{equation*}
\frac{L_{2}}{π} z^{2} - \left(\frac{1}{2}c_{5}+\frac{2L_{2}c_{1}}{π}\right)z + c_{1}c_{5} - \frac{L_{1}c_{2}}{π} \sim z^{2}-α^{2},
\end{equation*}
thus
\begin{align}
c_{5} & = - \frac{4L_{2}}{π} c_{1}, \label{eq:MPcoeff3}\\
4 c_{1}^{2} & = α^{2}-\frac{L_{1}}{L_{2}}c_{2}^{\vphantom{1}}. \label{eq:MPcoeff4}
\end{align}
These four equations allow us to express $c_{1}$, $c_{2}$, $c_{4}$ and $c_{5}$ in terms of $c_{3}$ (the sign of $c_1$ is fixed by (\ref{eq:MPcoeff9})). 

The coefficients $c_{7}$ and $c_{8}$ can be fixed by the minor which results from cancelling the second row and the first column
\begin{equation*}
\frac{2L_{1}}{π} z^{3} + \frac{c_{6}L_{1}}{π}z^{2} + \left(\frac{c_{7}L_{1}}{π}+\frac{2c_{2}L_{2}}{π}\right) z + \frac{L_{1}c_{8}}{π}-c_{2}c_{5}  
\sim z^{3} + bz^{2} - α^{2} z - b α^{2},
\end{equation*}
where $b$ is some constant. Again the ratios of the coefficients for the linear over the cubic and the constant over the quadratic term give
\begin{align}
c_{7} & = - 2 α^{2} - \frac{2L_{2}}{L_{1}}c_{2}, \label{eq:MPcoeff5} \\
c_{8} & = -α^{2} c_{6} + \frac{π}{L_{1}} c_{2}c_{5}. \label{eq:MPcoeff6}
\end{align}
The last coefficient that remains undetermined is $c_{6}$, but the determinant is going to help us for this. The requirement $\det P = 1$ implies
\begin{align*}
\left(z^{2} -α^{2}\right)^{3} & = z^{6} + \left(\frac{1}{2}c_{6} - 2c_{1}-c_{3}\right) z^{5} \\
							& \hspace{0.4cm} + \left(2c_{1}c_{3}-c_{1}c_{6}-c_{4}-\frac{1}{2}c_{3}c_{6}
+\frac{1}{2}c_{7}\right)z^{4} + …
\end{align*}
The quintic term immediately gives the desired expression
\begin{equation}\label{eq:MPcoeff7}
c_{6} = 4c_{1}+2c_{3}. 
\end{equation}
Exploiting furthermore the quartic term we get 
\begin{equation*}
-3α^{2} = 2 c_{1}c_{3}- c_{1}c_{6} - c_{4} - \frac{1}{2}c_{3}c_{6}+\frac{1}{2}c_{7},
\end{equation*}
which, by using the above obtained relations, is equivalent to
\begin{equation} \label{eq:MPcoeff9}
α^{2} = 4c_{1}^{2} + 4 c_{1}^{\vphantom{1}}c_{3}^{\vphantom{1}} + c_{3}^{2} + \frac{L_{2}}{L_{1}}c_{2}^{\vphantom{1}}.
\end{equation}
Let us relabel the parameters in accordance with \cite{Harmark:2004rm} as follows
\begin{equation*}
c_{3}^{\vphantom{1}} = \frac{1}{2} ρ_{0}^{2}, \quad L_{1}^{\vphantom{1}} = \frac{π}{4} a_{1}^{\vphantom{2}}ρ_{0}^{2},  \quad L_{2}^{\vphantom{1}} 
= \frac{π}{4} a_{2}^{\vphantom{2}}ρ_{0}^{2}.
\end{equation*}
Note that from the asymptotic patching matrix we see that $c_{3}$ is proportional to the mass which justifies the implicit assumption about its 
positiveness in the above definition. However, the parameters $(ρ_{0},a_{1},a_{2})$ are not unconstrained as we will see soon.

By \eqref{eq:MPcoeff1} we have
\begin{equation} \label{eq:MPcoeff10}
c_{2}^{\vphantom{1}} = - \frac{1}{4} a_{1}^{\vphantom{1}}a_{2}^{\vphantom{1}}ρ_{0}^{2}.
\end{equation}
Equations~\eqref{eq:MPcoeff4}, \eqref{eq:MPcoeff9}, \eqref{eq:MPcoeff10} imply
\begin{equation*}
\frac{L_{1}}{L_{2}} c_{2}^{\vphantom{1}} = 4 c_{1}^{\vphantom{1}} c_{3}^{\vphantom{1}} + c_{3}^{2} + \frac{L_{2}}{L_{1}} c_{2}^{\vphantom{1}} 
\quad ⇒ \quad c_{1}^{\vphantom{1}} = - \frac{1}{8} \left(ρ_{0}^{2}+a_{1}^{2}-a_{2}^{2}\right).
\end{equation*}
Moreover, from \eqref{eq:MPcoeff3} and \eqref{eq:MPcoeff7} we obtain
\begin{equation*}
c_{5}^{\vphantom{1}} = \frac{1}{8} a_{2}^{\vphantom{1}} ρ_{0}^{2} \left(ρ_{0}^{2}+a_{1}^{2}-a_{2}^{2}\right) \quad \text{and} \quad c_{6}^{\vphantom{1}} 
= \frac{1}{2} \left(ρ_{0}^{2}-a_{1}^{2}+a_{2}^{2}\right).
\end{equation*}
Continuing with \eqref{eq:MPcoeff2} yields
\begin{equation*}
c_{4}^{\vphantom{1}} = α^{2}-\frac{1}{8} ρ_{0}^{2}\left(ρ_{0}^{2}+a_{1}^{2}-a_{2}^{2}\right),
\end{equation*}
and \eqref{eq:MPcoeff5} and \eqref{eq:MPcoeff6} give 
\begin{align*}
c_{7}^{\vphantom{1}} & = -2α^{2} + \frac{1}{2} a_{2}^{2} ρ_{0}^{2}, \\
c_{8}^{\vphantom{1}} & = \frac{1}{2} α^{2} \left(-ρ_{0}^{2}+a_{1}^{2}-a_{2}^{2}\right) - \frac{1}{8} a_{2}^{2}ρ_{0}^{2}\left(ρ_{0}^{2}+a_{1}^{2}-a_{2}^{2}\right).
\end{align*}
With these parameters being determined and with the help of \eqref{eq:MPcoeff9} we can write $α$ explicitly as 
\begin{equation} \label{eq:MPalpha}
α^{2} = \frac{1}{16}\left(ρ_{0}^{2}-a_{1}^{2}-a_{2}^{2}\right)^{2}-\frac{1}{4}a_{1}^{2}a_{2}^{2}.
\end{equation}
Comparing those expressions with \eqref{eq:MPpatmat} one will find that they coincide. However, note that
\begin{equation*}
16α^{2} = ρ_{0}^{4}-2ρ_{0}^{2}\left(a_{1}^{2}+a_{2}^{2}\right)+\left(a_{1}^{2}-a_{2}^{2}\right)^{2},
\end{equation*}
which implies that for real non-zero $α$ we need the left hand side to be positive and therefore we need either $ρ_{0}^{2}>\left(|a_{1}^{\vphantom{1}}|+
|a_{2^{\vphantom{1}}}|\right)^{2}$, a condition on the asymptotic quantities familiar from the discussion of the Myers-Perry solution in 
\cite{Emparan:2008aa} and \cite{Myers:2011yc}, or $0<ρ_{0}^{2}<\left(|a_{1}^{\vphantom{1}}|-|a_{2}^{\vphantom{1}}|\right)^{2}$. 
This latter possibility is ruled out by (\ref{eq:MPcoeff0}): we want $\alpha>2c_1$ while (\ref{eq:MPcoeff4}) gives $\alpha^2\leq 4c_1^2$, so we must have $c_1<0$ or
\begin{equation*}
\rho_{0}^{2}\geq a_1^2-a_2^2;
\end{equation*}
from the bottom rod we must obtain this condition with $a_1$ and $a_2$ interchanged, so that we require
\begin{equation}
\rho_{0}^{2}\geq |a_1^2-a_2^2|,
\end{equation}
and this is incompatible with $0<ρ_{0}^{2}<\left(|a_{1}^{\vphantom{1}}|-|a_{2}^{\vphantom{1}}|\right)^{2}$ and nondegeneracy.


Mass and angular momenta form a set of three parameters and the position of the nuts can be expressed in terms of these three parameters. 
This is more than one would have expected just from Theorem~\ref{thm:holluniqueness}. However, we stated already that by 
\cite[Prop.~3.1]{Chrusciel:2011eu} the mass is redundant in the set of parameters. Here we did not eliminate the mass, but rather the 
rod length. If one instead replaces $M$ by $α$ in the set of parameters, one obtains an equation for $M$: by rearranging \eqref{eq:MPalpha} 
one seeks positive $c_{3}$ which satisfy a $6^{\mathrm{th}}$ order polynomial. With no 
further conditions on $(α>0, L_{1}^{\vphantom{1}},L_{2}^{\vphantom{1}})$ there are again two positive solutions for $c_{3}$ 
(unless $L_{1}^{2}=L_{2}^{2}$ when there is only one) but once again one branch is ruled out by (\ref{eq:MPcoeff1}).

Some of the steps above, when we determined all the parameters in the patching matrix, required $L_{1}L_{2}≠0$. Assuming that one 
of the angular momenta vanishes leads to dichotomies at certain steps when solving for the $c_i$. Some of the branches in this 
tree of possibilities lead to contradictions while others lead to valid solutions such as the Myers-Perry solution with one vanishing angular momentum or 
an ultrastatic solution, that is where $g_{tt}=1$, $g_{ti}=0$ (which must violate one of the global conditions as the mass is zero). On the other hand at no point did we use the fact that the middle rod is a horizon.

Note also that issues of conicality cannot arise here as the periodicities of $\phi,\psi$ are chosen to be $2\pi$ on the outer parts of 
the axis and no further spatial rods are left. When we turn to a larger numbers of nuts there could be conical singularities.\\ \mbox{} \hfill $\blacksquare$
\end{ex}
Moving on to a rod structure with three nuts, we will consider the simpler case where one of the Killing vectors 
is hypersurface-orthogonal.  
\begin{ex}[Three-Nut Rod Structure with one Hypersurface-Orthogonal Killing Vector]\mbox{}\\
We consider the rod structure as in Figure~\ref{fig:threenutrodstr}. 
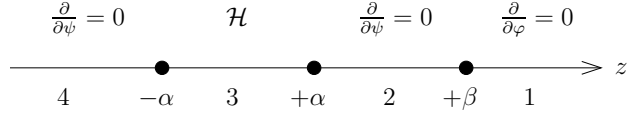
\begin{figure}[htbp]
\begin{center}
     \scalebox{0.8}{\input{figure3.tex}}
     \caption{Rod structure with three nuts, where $α,β>0$, and $S^{2}×S^{1}$ horizon.} 
     \label{fig:threenutrodstr}
\end{center}
\end{figure}
Together with $L_{1}=L\neq 0$, $L_{2}=0$ this comprises our twistor data. In order to simplify the calculations we would like to make assumptions 
such that the two non-diagonal entries in the third row and column of the patching matrix vanish (when adapted to $(β,∞)$). One therefore needs 
$g_{t{\psi}}=g_{{\varphi}{\psi}}=0$. This cannot be concluded from $L_{2}=0$, as the Black Saturn shows (see \cite{Elvang:2007rd}). We thus make the assumption that $\partial_{\psi}$ is 
hypersurface-orthogonal, that is $Ψ∧\drm Ψ=0$, so that $g_{t{\psi}}=g_{{\varphi}{\psi}}=0$, and $χ_{2}=0$. 

These assumptions turn our ansatz into
\begin{equation*}
\renewcommand{\arraystretch}{1.5}
P = \frac{1}{Δ}
\left( \begin{array}{ccc}
q(z) & l(z) & 0 \\
l(z) & c(z) & 0 \\
0 & 0 & Q(z)
\end{array}\right),
\end{equation*}
where
\begin{equation*}
\renewcommand{\arraystretch}{2}
\begin{array}{rcl}
Δ (z)& = & (z+α)(z-α)(z-β), \\
q(z) & = & \dfrac{1}{2} z^{2} + c_{1}z+c_{2}, \\
l(z) & = &\dfrac{L}{π} z + c_{3}, \\
c(z) & = &-z^{3} + c_{4} z^{2} + c_{5} z +c_{6}, \\
Q(z) & = & 2 z^{4} + c_{7} z^{3} + c_{8} z^{2} + c_{9} z + c_{10}.
\end{array}
\end{equation*}
Corollary~\ref{cor:invpmatrix} gives the following conditions
\begin{equation} \label{eq:BRcoeff1}
\renewcommand{\arraystretch}{1.4}
\begin{array}{rcrrl}
qc-l^{2} & = & \tilde q_{1}Δ,\qquad & \tilde q_{1} & \hspace{-0.2cm} \text{ quadratic},\\
Qq^{\hphantom{2}} & = & \tilde c_{1}Δ,\qquad & \tilde c_{1} & \hspace{-0.2cm} \text{ cubic},\\
Ql^{\hphantom{2}} & = & \tilde q_{2}Δ,\qquad & \tilde q_{2} & \hspace{-0.2cm} \text{ quadratic},\\
Qc^{\hphantom{2}} & = & \tilde Q_{1}Δ,\qquad & \tilde Q_{1} & \hspace{-0.2cm} \text{ quartic}.
\end{array}
\end{equation}
The condition for the patching matrix to have unit determinant then implies
\begin{equation}\label{eq:BRcoeff2}
Δ^{3} = Q (qc-l^{2}) = Q \tilde q_{1} Δ \quad ⇔ \quad Δ^{2} = Q \tilde q_{1}. 
\end{equation}
Now, as $\tilde q_{1}$ is a quadratic, there are six possibilities for it to be a product of $(z+α)$, $(z-α)$ and $(z-β)$. But 
$∂_{ψ}=0$ on $(α,β)$, thus $Q/Δ→0$ for $z \downarrow β$. To guarantee this $(z-β)^{2}$ has to divide $Q$, which rules out three 
of those six possibilities. Furthermore, by Corollary~\ref{cor:invpmatrix} we have 
\begin{equation*}
\frac{\tilde q_{1}}{Δ \vphantom{\skew{7}{\tilde}{A}_{4}}}=\frac{1}{\det \skew{7}{\tilde}{A}_{4}} \quad \text{on } (-∞,-α),
\end{equation*}
where $\skew{7}{\tilde}{A}_{4}$ is obtained from $J$ by cancelling the rows and columns containing inner products with $∂_{ψ}$. But 
from the general theory we know that the entry of $P$ with the inverse determinant contains a simple pole when approaching the nut, 
that is $z \uparrow -α$, so that $\tilde q_{1}(-α) ≠ 0$. This immediately yields
\begin{equation*}
\tilde q_{1} = \frac{1}{2} (z-α)^{2} \quad \text{and by \eqref{eq:BRcoeff2} also} \quad Q = 2 (z+α)^{2}(z-β)^{2}.
\end{equation*}
Now observe that there is a factor of $(z-α)$ in $Δ$ but not in $Q$, so that by \eqref{eq:BRcoeff1} the monic $(z-α)$ has to divide 
$l$, $q$ and $c$. We write this as
\begin{equation*}
l = \frac{L}{π} (z-α), \quad q = -\frac{1}{2} (z-α)\, \tilde l_{1}, \quad c = - (z-α)\, \tilde q_{3},
\end{equation*}
where
\begin{equation*}
\tilde l_{1} = z+A, \quad \tilde q_{3} = z^{2} + Bz + C \qquad \text{for } A, B, C = \text{const.}
\end{equation*}
The first equation in \eqref{eq:BRcoeff1} then turns into
\begin{equation*} 
\renewcommand{\arraystretch}{2}
\begin{array}{crcl}
& \tilde l_{1} \tilde q_{3} - \dfrac{2L^{2}}{π^{2}} & = & Δ \\
⇔ & z^{3} + (A+B) z^{2} + (C+AB) z + AC-\dfrac{2L^{2}}{π^{2}} & = & z^{3}-βz^{2}-α^{2}z+α^{2}β.
\end{array}
\end{equation*}
Comparing the coefficients one sees
\begin{equation*}
B = -A-β, \qquad C+AB=-α^{2}, \qquad AC-\frac{2L^{2}}{π^{2}} = α^{2}β,
\end{equation*}
and therefore $A$ satisfies
\begin{equation*}
\renewcommand{\arraystretch}{2}
\begin{array}{crcl}
& \dfrac{1}{A} \left(α^{2}β+\dfrac{2L^{2}}{π^{2}}\right) - A(A+β) & = & - α^{2} \\
⇔ & A^{3} + βA^{2}-α^{2}A-α^{2}β - \dfrac{2L^{2}}{π^{2}} & = & 0.
\end{array}
\end{equation*}
Writing $F(a) \coloneqq a^{3} + βa^{2}-α^{2}a-α^{2}β - \dfrac{2L^{2}}{π^{2}}$, we see that since $F(0)<0$, this last polynomial has to 
have at least one (positive) real root which we'll call $A$ (see Figure~\ref{fig:poly}). 
\begin{figure}[htbp]
\begin{center}
     \scalebox{0.9}{\input{figure4.tex}}
     \caption{The cubic $F(a)$.} 
     \label{fig:poly}
\end{center}
\end{figure}
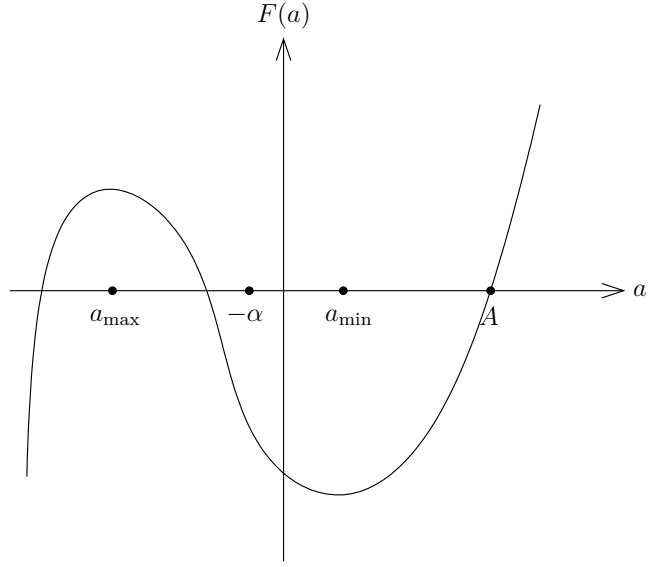
Now from $F'(A)=3A^{2}+2βA-α^{2}$ one concludes that the local maximum of $F$ is at
\begin{equation*}
a_{\mathrm{max}} = -\frac{1}{3} \left(β+\sqrt{β^{2}+3α^{2}}\right). 
\end{equation*}
Furthermore, note that since $α≤β$, we have
\begin{align*}
a_{\mathrm{max}} & ≤ -\frac{1}{3} \left(α+\sqrt{α^{2}+3α^{2}}\right) = -α \quad \text{and} \\
F(-α) & = -\dfrac{2L^{2}}{π^{2}}<0,
\end{align*}
which implies that if $F$ has two more real roots, they will both be smaller than $-α$. On the other hand there is a constraint on 
$A$ obtained from the asymptotics. In our patching matrix the central entry is
\begin{equation*}
\frac{c}{Δ} = - \frac{(z-α)(z^{2}+Bz+C)}{Δ} = - 1 + (α-β-B) z^{-1} + …
\end{equation*}
Using \eqref{eq:asymptPtop} and the relation between $A$ and $B,$ this gives
\begin{equation*}
A+α = \frac{4M}{3π}.
\end{equation*}
Positivity of $M$ thus implies $A>-α$ and we therefore have shown that there is a unique positive $A∈ℝ$ which satisfies all the constraints.

Consequently, by our ansatz we are able to fix all the parameters in terms of $α$, $β$, $L$, that is in terms of the given rod and asymptotic data, and the patching matrix is 
\begin{equation*} 
\renewcommand\arraystretch{2.5}
P_{1}=\left(\begin{array}{ccc}-\dfrac{z+A}{2(z+α)(z-β)} & \dfrac{L}{π(z+α)(z-β)} & 0 \\
\cdot  & -\dfrac{z^{2}-\skew{2}{\tilde}{γ} z+\skew{3}{\tilde}{δ}}{(z+α)(z-β)} & 0 \\
0 & 0 & \dfrac{2(z+α)(z-β)}{z-α}\end{array}\right),
\end{equation*}
where
\begin{equation*}
\skew{2}{\tilde}{γ} = β+A, \quad \skew{3}{\tilde}{δ}=-α^{2}+βA +A^{2}.
\end{equation*}
Now compare this with \eqref{eq:BRpatmat}: since $λ$ and $A$ are zeros of the same polynomial and are restricted by the same inequality involving the mass, they are equal and we have derived the patching matrix 
for the black ring with the conical singularity not yet removed. (We are grateful to Harvey Reall for suggesting this possibility.)  \\ \mbox{}\hfill $\blacksquare$

For the regular black ring, removing the 
conical singularity gives the angular momentum $L$ in terms of $α$ and $β$. In this formalism, removing the conical singularity requires more work which we turn to next.
\end{ex}

\subsection{Local Behaviour of $J$ around a Nut} \label{subsec:Jaroundnut}

For the case of a rod structure with three nuts and $L_1L_2\neq 0$, and generally as the number of nuts gets higher, one needs more constraints and these will come from the inner rods. It is therefore important to have an understanding of how the patching matrices with adaptations to adjacent rods are related to each other. We have seen an 
example in Theorem~6.5 in Part~I, which can be considered as such a switch at the nut at infinity. The proof gives an idea of what is happening 
when changing the adaptation, yet it will be more difficult for interior nuts, that is nuts for which $|a_{i}|$ is finite. 

A strategy of how to achieve this is described in \cite[Ch.~3]{Fletcher:1990aa}. There the essence is that ``... redefining the sphere $S_{0}$ and $S_{1}$ 
by interchanging double points alters the part of the real axis to which the bundle is adapted.'' \cite[Sec.~3.2]{Fletcher:1990aa}. However, as the example 
in \cite[Sec.~5.1]{Fletcher:1990aa} shows, this comes down to a Riemann-Hilbert problem which will be rather hard and impractical to solve in five or even 
higher dimensions. Thus we will approach this task in a different way. The idea is that we start off as above on the outermost rods where $|z| \to \infty $, 
determine as many free parameters as possible by the constraints which we have got on these rods, then take the resulting $P$-matrix (still having free 
parameters in it which we would like to pin down), calculate its adaptation to the next neighbouring rod and apply analogous constraints there. But before 
looking at the patching matrix itself let us first study how $J$ behaves locally around a nut.

Consider first a nut where two spatial rods meet, that is like in Figure~\ref{fig:sprods}.\vspace{0.3cm}
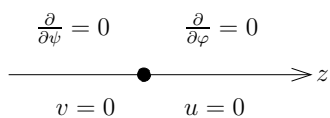
\begin{figure}[htbp]
\begin{center}
     \scalebox{0.8}{\input{figure5.tex}}
     \caption{Two spatial rods with their rod vectors meeting at a nut.} 
     \label{fig:sprods}
\end{center}
\end{figure}

Without loss of generality assume that the nut is at $z=0$. In this case a suitable choice of coordinates are the $(u,v)$-coordinates defined as 
\begin{equation*} 
r=uv, \  z=\frac{1}{2}(v^{2}-u^{2}) \quad \Leftrightarrow  \quad u^{2} = -z\pm \sqrt{r^{2}+z^{2}},\ v^{2} = z\pm \sqrt{r^{2}+z^{2}},
\end{equation*}
where the signs on the right-hand side are either both plus or both minus. If we choose both signs to be plus, then the rod $∂_{\varphi}=0$ corresponds 
to $u=0$ and $∂_{\psi}$ to $v=0$. The metric in the most general case has the form
\begin{equation} \label{eq:GCmetric} 
\begin{split}
\drm s^{2} & = X \,\drm t^{2} + 2Y \,\drm t \drm {\varphi} + 2Z \,\drm t \drm {\psi} + U \,\drm {\varphi}^{2} + 2V \,\drm {\varphi} \drm {\psi} + W \,\drm {\psi}^{2}\\
& \hphantom{=} + {\erm}^{2ν} (u^{2}+v^{2})(\drm u^{2}+\,\drm v^{2}),
\end{split}
\end{equation}
or equivalently
\begin{equation*}
\renewcommand{\arraystretch}{1.4}
J (u,v)=\left(\begin{array}{ccc}X & Y & Z \\\cdot  & U & V \\\cdot  & \cdot  & W\end{array}\right).
\end{equation*}
We assume that $\phi,\psi$ have period $2\pi$.

\begin{thm}\label{thm:Jaroundnut}
For a space-time regular on the axis the generic form of $J$ in $(u,v)$-coordinates around a nut, where two spacelike rods meet, is
\begin{equation} \label{eq:Jaroundnut}
\renewcommand{\arraystretch}{1.4}
J= \left(\begin{array}{ccr}
X_{0} & u^{2}Y_{0} & v^{2}Z_{0} \\
\cdot  & u^{2}U_{0} & u^{2}v^{2} V_{0} \\
\cdot  & \cdot  & v^{2}W_{0}\end{array}\right),
\end{equation}
and, furthermore, one needs
\begin{itemize}
\item $\dfrac{U_{0}^{\hphantom{1}}}{v^{2}_{\hphantom{1}}\erm^{2\nu}}=1$ as a function of $v$ on $u=0$,
\item $\dfrac{W_{0}^{\hphantom{1}}}{u^{2}_{\hphantom{1}}\erm^{2\nu}}=1$ as a function of $u$ on $v=0$.
\end{itemize}

If one of the rods is the horizon instead of a spacelike rod corresponding statements hold.
\end{thm}

The second part of the theorem is closely tied to the problem of conicality, which we will investigate shortly.

\begin{proof}
Introduce Cartesian coordinates
\begin{equation} \label{eq:cartaroundnut}
x=u\cos\phi,\quad y=u\sin\phi,\quad z=v\cos\psi,\quad w=v\sin\psi,
\end{equation}
then the metric becomes in these coordinates
\begin{equation} \label{eq:cartesmet} 
\begin{split}
\drm s^{2} & = X\, \drm t^2+2\frac{Y}{u^2}\, \drm t(x\, \drm y-y\, \drm x)+2\frac{Z}{v^2}\, \drm t(z\, \drm w-w\, \drm z) +\frac{U}{u^4}(x\, \drm y-y\, \drm x)^2\\
& \hspace{0.4cm}+2\frac{V}{u^2v^2}(x\, \drm y-y\, \drm x)(z\, \drm w-w\, \drm z)+\frac{W}{v^4}(z\, \drm w-w\, \drm z)^2 \\
& \hspace{0.4cm} + \erm^{2\nu}(u^2+v^2)\left(\frac{1}{u^2}(x\,\drm x+y\, \drm y)^2+\frac{1}{v^2}(z\, \drm z+w\, \drm w)^2\right).
\end{split}
\end{equation}
The $x$, $y$, $z$, $w$ are not to be confused with the earlier use of the same symbols. Set $X_{0}=X$. Now as $u→0$ for constant $v$ we immediately 
see that in order for $g_{ty}$ and $g_{xw}$ to be bounded we need $Y=u^{2}Y_{0}$ and $V=u^{2}V_{1}$ for bounded $Y_{0}$, $V_{1}$. The remaining singular terms are
\begin{equation*}
\frac{U}{u^4}(x\, \drm y-y\, \drm x)^2+ \erm^{2\nu}(u^2+v^2)\frac{1}{u^2}(x\, \drm x+y\, \drm y)^2.
\end{equation*}
For the fourth-order pole not to be dominant we need $U=u^{2}U_{0}$ for bounded $U_0$; then it is required
\begin{equation} \label{eq:conf1}
\frac{U_{0}}{v^{2}\erm^{2\nu}}=1 \quad \text{as a function of } v \text{ on } u=0
\end{equation}
to remove the remaining second-order pole.

Repeating this for $v→0$ with fixed $u$ yields $Z=v^{2}Z_{0}$, $V_{1}=v^{2}V_{0}$, $W=v^{2}W_{0}$ and 
\begin{equation*}
\frac{W_{0}}{u^{2}\erm^{2\nu}}=1 \quad \text{as a function of } u \text{ on } v=0.
\end{equation*}
This is the minimum that we can demand in terms of regularity of $J$ on the axis and near the nuts.

Assuming now without loss of generality that in Figure~\ref{fig:sprods} the axis segment where $v=0$ is the horizon, we have seen in Section~5 in Part~I 
that then the first row and first column degenerate. So, we substitute
\begin{equation*}
z=v\cosh(ωt),\quad w=v\sinh (ωt),
\end{equation*}
where $ω$ is a constant with no further restriction. The coordinates $x$ and $y$ choose as in \eqref{eq:cartaroundnut}. Now the above argument works analogously 
with all results equivalent, but
\begin{equation*}
\frac{X_{0}}{v^{2}\erm^{2\nu}}=-ω^{2} \quad \text{as a function of } v \text{ on } u=0.
\end{equation*}
\end{proof}

\subsection{Conicality and the Conformal Factor}
Returning to the case as depicted in Figure~\ref{fig:sprods}, we saw in \eqref{eq:conf1} that regularity at an axis seqment where $\partial_{\phi}$ vanishes forces 
a relation between $g_{\phi\phi}$ and the conformal factor $\erm^{2\nu}$ of the $(r,z)$-metric. In this section we first establish the following.

\begin{prop} \label{prop:conf1}
On a segment of the axis where $u=0$ we have $\frac{U_0}{v^2\erm^{2\nu}}=\text{constant}$.  
\end{prop}
\begin{proof}
To prove this we need to consider how the conformal factor varies on the axis and this is obtained from the second part of the Einstein field equations
\begin{equation} \label{eq:conf2}
\partial_{\xi}\left(\log \left(r\erm^{2\nu}\right)\right)=\frac{\irm r}{2} \tr\left(J^{-1}J_{\xi} J^{-1}J_{\xi}\right).
\end{equation}
It will be convenient to work with $\chi=u+\irm v$ where $\xi=z+\irm r=\frac{1}{2}\chi^{2}$ and concentrate on the conformal factor of the $(u,v)$-metric which 
is $(u^{2}+v^{2})\erm^{2\nu}$ by \eqref{eq:cartesmet}. Then
\begin{equation*} 
\begin{split}
\partial_{\chi}\left(\log\left(\left(u^{2}+v^{2}\right)\erm^{2\nu}\right)\right) & 
=\partial_{\chi}\left(\log\left(\left(u^{2}+v^{2}\right)(uv)^{-1}r\erm^{2\nu}\right)\right) \\
& =\frac{1}{\chi}-\frac{1}{2u}+\frac{\irm}{2v}+\frac{\irm uv}{2(u+\irm v)}\tr \left(J^{-1}J_{\chi} J^{-1}J_{\chi}\right).
\end{split}
\end{equation*}

Close to the axis segment $u=0$ we substitute from \eqref{eq:Jaroundnut} and expand in powers of $u$ to find
\begin{equation} \label{eq:conf3} 
\partial_{\chi}\left(\log\left(\left(u^{2}+v^{2})\erm^{2\nu}\right)\right)\right)=\frac{1}{\chi}-\frac{1}{2u}+\frac{\irm}{2v}
+\frac{\irm uv}{2(u+\irm v)}\left(\frac{K_{1}}{u^{2}}+\frac{K_{2}}{u}+\mathcal{O}(1)\right),
\end{equation}
where
\begin{align*}
K_{1} & =\left(U_{0}\left(X_{0}W_{0}-v^{2}Z_{0}^{2}\right)\right)^{2}=1+\frac{1}{v^2}\mathcal{O}\left(u^{2}\right), \\
K_{2} & =\left(U_{0}\left(X_{0}W_{0}-v^{2}Z_{0}^{2}\right)\right)^{2}\frac{\partial_{\chi} U_{0}}{U_{0}}.
\end{align*}
The right hand side of the first equation follows from the determinant
\begin{equation*}
u^{2}v^{2}=\det J = u^{2}v^{2} X_{0}U_{0}W_{0} - u^{2}v^{4}U_{0}Z_{0}^{2} + \mathcal{O}\left(u^{4}\right).
\end{equation*}
Taking in \eqref{eq:conf3} the limit on to $u=0$ we obtain just
\begin{equation*}
\partial_{v}\left(\log\left(v^{2}\erm^{2\nu}\right)\right)=\partial_{v}\log\left(U_{0}\right),
\end{equation*}
so that
\begin{equation*}
\frac{U_{0}}{v^{2}\erm^{2\nu}}=\text{constant} \quad  \text{on } u=0.
\end{equation*}
\end{proof}

Thus \eqref{eq:conf1} will hold at all points of the axis segment if it holds at one. The following proposition is an analysis similar to 
\cite[App.~H]{Harmark:2004rm}, but it is simpler and more self-contained to rederive it than translate it.
\begin{prop} \label{prop:conf2}
As a function on the axis $\{u=0\}\cup\{v=0\}$, that is as a function of one variable, the factor $\left(u^{2}+v^{2}\right)\erm^{2\nu}$ is 
continuous at the nut $u=v=0$.
\end{prop}
\begin{proof}
Near the nut introduce polar coordinates
\begin{equation*}
u=R\cos\Theta, \quad v=R\sin\Theta,
\end{equation*}
so that from \eqref{eq:conf2} we get
\begin{align*} 
\partial_{\Theta}\left(\log\left(\left(u^{2}+v^{2}\right)\erm^{2\nu}\right)\right) 
& =\left(u\partial_{v}-v\partial_{u}\right)\left(\log\left(\left(u^{2}+v^{2}\right)\erm^{2\nu}\right)\right) \\ 
& =-\frac{u}{v}+\frac{v}{u}-\frac{uv}{4}\tr\left(J^{-1}J_{u} J^{-1}J_{u}-J^{-1}J_{v} J^{-1}J_{v}\right).
\end{align*}
Again we expand this using \eqref{eq:Jaroundnut} to find
\begin{align*} 
\partial_{\Theta}\left(\log\left(\left(u^{2}+v^{2}\right)\erm^{2\nu}\right)\right) & =-\frac{u}{v}+\frac{v}{u}
-\frac{v}{u}\left(U_{0}\left(X_{0}W_{0}-v^{2}Z_{0}^{2}\right)\right)^{2}\\ 
& \hspace{0.4cm}+\frac{u}{v}\left(W_{0}\left(X_{0}U_{0}-u^{2}Y_{0}^{2}\right)\right)^2+\mathcal{O}(u)+\mathcal{O}(v) \\ 
& = \mathcal{O}(u)+\mathcal{O}(v)=\mathcal{O}(R).
\end{align*}
Now the jump in $\log\left(u^{2}+v^{2}\right)\erm^{2\nu}$ round the nut is 
\begin{equation*} 
\Delta\left(\log\left(\left(u^{2}+v^{2}\right)\erm^{2\nu}\right)\right)=\lim_{R\rightarrow 0}\int_{0}^{\frac{π}{2}}\partial_{\Theta}\left(\log\left(\left(u^{2}
+v^{2}\right)\erm^{2\nu}\right)\right)\drm \Theta=0,
\end{equation*}
and $\left(u^{2}+v^{2}\right)\erm^{2\nu}$ does not jump either.

On $u=0$, $U_{0}$ is continuous and by Proposition~\ref{prop:conf1} $\frac{U_0}{v^2\erm^{2\nu}}=\text{constant}$, so $v^2\erm^{2\nu}$ must be bounded there. 
Similarly, on $v=0$ for $W_{0}$ and $\frac{W_{0}}{u^{2}\erm^{2\nu}}$. Thus, $\left(u^{2}+v^{2}\right)\erm^{2\nu}$ is continuous on the  two rods and has no 
jump across the nut, so it is continuous on the axis.
\end{proof}

The strategy for removing conical sigularities is now clear: We start by assuming that $\phi$ and $\psi$ both have period $2\pi$. On the part of the axis 
extending to $z=+\infty$, where the Killing vector $\partial_{\phi}$ vanishes, we have $ \frac{U_{0}}{v^{2}\erm^{2\nu}}=\text{constant}$ by 
Proposition~\ref{prop:conf1} and the asymptotic conditions we are imposing make this constant one. The corresponding statement holds on the part of 
the axis extending to $z=-\infty$ for the same reason. When passing by a nut between two space-like rods we may suppose, by choosing the basis of 
Killing vectors appropriately, that $\partial_{\phi}$ vanishes above the nut and $\partial_{\psi}$ below and we know by Proposition~\ref{prop:conf2} 
that $\left(u^2+v^2\right)e^{2\nu}$ is continuous at the nut. If there is no conical singularity above the nut we have $ \frac{U_{0}}{v^{2}\erm^{2\nu}}=1$ 
there and we want $ \frac{W_{0}}{u^{2}\erm^{2\nu}}=1$ below the nut. Therefore we require the limits of $U_{0}$ from above and $W_{0}$ from below to be equal.
\begin{cor} \label{cor:conic}
With the conventions leading to \eqref{eq:Jaroundnut}, the absence of conical singularities requires
\begin{equation*}
\lim_{v\rightarrow 0}U_{0}=\lim_{u\rightarrow 0}W_0.
\end{equation*}
\end{cor}
This is what we have just shown. At a nut where one rod is the horizon we do not obtain further conditions as we have no reason to favour a particular 
value of $ω$. To see how this is applied to the case of the black ring we need a better understanding of going past a nut.

\subsection{Local Behaviour of $P$ around a Nut: Switching}

In this section we establish a prescription for obtaining the matrix $P_{-}$ adapted to the segment of the axis below a nut from the matrix $P_{+}$ adapted 
to the segment above. We call this process `switching'. Once we have the prescription we can impose the condition of non-conicality found in Corollary~\ref{cor:conic}. 
We then apply this to the black ring, but it is clear that with this prescription we have an algorithm for working systematically down the axis given any rod structure 
so that we obtain all the matrices $P_{i}$ adapted to the different rods labelled by $i$. The result is the following.
\begin{thm} \label{thm:switching}
Let $z=a$ be a nut where two spacelike rods meet, as in Figure~\ref{fig:sprods}, and assume that we have chosen a gauge where the twist potentials vanish when 
approaching the nut. Then 
\begin{equation*} 
\renewcommand{\arraystretch}{1.5}
P_{-}^{\vphantom{\frac{1}{2}}}=\left(\begin{array}{ccc}0 & 0 & \dfrac{1}{2(z-a)} \\0 & 1 & 0 \\ 2(z-a) & 0 & 0\end{array}\right)P_{+}^{\vphantom{\frac{1}{2}}}\left(\begin{array}{ccc}0 & 0 & 2(z-a) \\0 & 1 & 0 \\\dfrac{1}{2(z-a)} & 0 & 0\end{array}\right),
\end{equation*}
where $P_{+}$ is adapted to $u=0$ and $P_{-}$ is adapted to $v=0$.
\end{thm}

We begin by motivating this prescription from a consideration of \eqref{eq:GCmetric}. First calculate the twist potentials in the same way as in Section~\ref{sec:Pexamples}. 
The metric \eqref{eq:GCmetric} can be rearranged in orthonormal form
\begin{equation*}
\begin{split}
\drm s^{2} & = X (\drm t +{\omega}_{1}\, \drm {\varphi} +{\omega}_{2}\,\drm {\psi})^{2} + \skew{2}{\tilde} U (\drm {\varphi}+ {\Omega} \, \drm {\psi})^{2} \\
& \hphantom{=} + \tilde W\, \drm ψ^{2} - \erm^{2{\nu}} (\drm r^{2}+\drm z^{2}).
\end{split}
\end{equation*}
The orthonormal frame is again
\begin{equation*}
\renewcommand{\arraystretch}{2}
\begin{array}{ll}{\theta}^{0} = X^{\frac{1}{2}}(\drm t + {\omega}_{1}\, \drm {\varphi} + {\omega}_{2} \, \drm {\psi}), & {\theta}^{1} 
=\skew{2}{\tilde} U ^{\frac{1}{2}} (\drm {\varphi}+{\Omega} \, \drm {\psi}), \\{\theta}^{2} = \tilde W^{\frac{1}{2}} \, \drm {\psi}, \hspace{0.6cm} {\theta}^{3} = 
\erm^{{\nu}} \, \drm r, & {\theta}^{4} = \erm^{{\nu}} \, \drm z,\end{array}
\end{equation*}
so
\begin{equation*}
\renewcommand{\arraystretch}{2}
\begin{array}{lr}X{\omega}_{1} = Y, \qquad X{\omega}_{2} = Z &  X{\omega}_{1}{\omega}_{2}+\skew{2}{\tilde}U {\Omega} = V, \\ \skew{2}{\tilde}U +X {\omega}_{1}^{2} 
= U, & \tilde W + \skew{2}{\tilde}U {\Omega}^{2}+X{\omega}_{2}^{2}= W.\end{array}
\end{equation*}
Adapted to $\partial_{\psi}=0$, then for small $r$ it is $Z,V,W\in \mathcal{O}(r^{2})$, hence ${\omega}_{2}, {\Omega}, \tilde W \in \mathcal{O}(r^{2})$, 
(In order to see that ${\Omega}\in \mathcal{O}(r^{2})$, derive from $X\in \mathcal{O}(1)$ and $\skew{2}{\tilde}U X = UX-Y^{2}$ that $\skew{2}{\tilde}U\in \mathcal{O}(1)$.) 
and the other terms $\mathcal{O}(1)$. This implies $\frac{\tilde W}{W} \to  1$ as $r→0$. 

Now the 1-forms are
\begin{equation*}
\partial _{t} \to  T = X^{\frac{1}{2}}{\theta}^{0}, \quad \partial _{{\varphi}} \to  {\Phi} = {\omega}_{1}X^{\frac{1}{2}}{\theta}^{0}+\skew{2}{\tilde}U^{\frac{1}{2}}{\theta}^{1},
\end{equation*}
hence
\begin{align*} 
\drm {\chi}_{1} & = * (T\wedge {\Phi}\wedge \drm T) = *(X^{\frac{1}{2}}{\theta}^{0}\wedge \skew{2}{\tilde}U^{\frac{1}{2}}{\theta}^{1}\wedge X\, \drm {\omega}_{2}\wedge \drm {\psi}) \\ 
\drm {\chi}_{2} & = *(T\wedge {\Phi}\wedge  \drm {\Phi}) = *(X^{\frac{1}{2}}{\theta}^{0}\wedge \skew{2}{\tilde}U^{\frac{1}{2}}{\theta}^{1}\wedge ({\omega}_{1}X\, \drm {\omega}_{2}+\skew{2}{\tilde}U\, \drm {\Omega})\wedge \drm {\psi}),
\end{align*}
which with $*({\theta}^{0}\wedge {\theta}^{1}\wedge {\theta}^{2}\wedge {\theta}^{3})={\epsilon}{\theta}^{4}$ leads to
\begin{align*}
\partial _{z} {\chi}_{1} & = {\epsilon} (X\skew{2}{\tilde}U)^{\frac{1}{2}}X \lim_{r\to 0}\left(\frac{\partial _{r}{\omega}_{2}}{\tilde W ^{\scriptscriptstyle{1/2}}_{\vphantom{0}}}\right) \\
\partial _{z} {\chi}_{2} & = {\epsilon} (X\skew{2}{\tilde}U)^{\frac{1}{2}} \lim_{r\to 0}\left(\frac{X{\omega}_{1}\partial _{r}{\omega}_{2}+\skew{2}{\tilde}U\partial _{r}{\Omega}}{W^{\scriptscriptstyle{1/2}}_{\vphantom{0}}}\right).
\end{align*}
We switch again to $(u,v)$-coordinates then on $v=0$ it is
\begin{equation*}
\frac{\partial {\chi}_{1}}{\partial u} = u \frac{\partial {\psi}_{1}}{\partial z} = u {\epsilon} (X\skew{2}{\tilde}U)^{\frac{1}{2}} X \lim_{v\to 0} \left(\frac{1}{u\tilde W^{\scriptscriptstyle{1/2}}_{\vphantom{0}}}\frac{\partial {\omega}_{2}}{\partial v}\right).
\end{equation*}
Now use 
\begin{equation*}
{\omega}_{2} = \frac{Z}{X} = \frac{v^{2}Z_{0}}{X_{0}} \quad \text{and} \quad \tilde W = v^{2}W_{0} + \mathcal{O}(v^{4}), 
\end{equation*}
to obtain
\begin{align*}
\frac{\partial {\chi}_{1}}{\partial u} & = 2{\epsilon}Z_{0}\left(\frac{U_{0}X_{0}}{W_{0}}\right)^{\frac{1}{2}}u+\mathcal{O}(u^{2})\\
\Rightarrow  {\chi}_{1}^{\vphantom{1}} & = {\chi}_{1}^{0} + {\chi}_{1}^{1} u^{2} + \text{h.o.} \vphantom{\left(\frac{U_{0}X_{0}}{W_{0}}\right)^{\frac{1}{2}}}
\end{align*}
Analogous steps lead to 
\begin{align*}
\frac{\partial {\chi}_{2}}{\partial u} & = 2{\epsilon} V_{0} \left(\frac{U_{0}X_{0}}{W_{0}}\right)^{\frac{1}{2}}u^{3}+ \text{h.o.} \\
\Rightarrow  {\chi}_{2}^{\vphantom{1}} & = {\chi}_{2}^{0} + {\chi}_{2}^{1} u^{4} + \text{h.o.} \vphantom{\left(\frac{U_{0}X_{0}}{W_{0}}\right)^{\frac{1}{2}}}
\end{align*}
For $u=0$ we only have to swap $Y\leftrightarrow Z$, $U\leftrightarrow W$. With $u^{2} \sim 2z$ the above can be summarized as 
\begin{equation*} 
\renewcommand{\arraystretch}{2}
P_{-}^{\vphantom{\frac{1}{2}}}(r=0,z) =\left(\begin{array}{ccr}\dfrac{g_{0}}{2z} + \mathcal{O}(1) & -g_{0}^{\vphantom{1}} {\chi}^{1}_{1}+\mathcal{O}(z) & -g_{0}^{\vphantom{1}} {\chi}^{1}_{2}z+\mathcal{O}(z^{2}) \\
\cdot  & \hspace{0.5cm}X_{0}+\mathcal{O}(z) & 2zY_{0}+\mathcal{O}(z^{3}) \\
\cdot  & \cdot  & 2zU_{0}+\mathcal{O}(z^{4})\end{array}\right),
\end{equation*}
where $g_{0}^{\vphantom{2}}=(X_{0}^{\vphantom{2}}U_{0}^{\vphantom{2}}-2zY_{0}^{2})^{-1}$. Note that here we dropped without loss of generality the constant terms of the 
twist potentials ${\chi}_{i}^{0}$. This can be done just by a gauge transformation to $P$ of the form $P\to APB$ with constant matrices $A$ and $B$, namely
\begin{equation*}
\renewcommand{\arraystretch}{1.4}
P \to  \left(\begin{array}{rrr}1 & \hphantom{-}0 & \hphantom{-}0 \\-c_{1} & 1 & 0 \\-c_{2} & 0 & 1\end{array}\right)\, P\, \left(\begin{array}{ccc}1 & -c_{1} & -c_{2} \\ 0 & \hphantom{-}1 & \hphantom{-}0 \\0 & \hphantom{-}0 & \hphantom{-}1\end{array}\right).
\end{equation*}
For $P$ in standard form this results in ${\chi}_{i}\to {\chi}_{i}+c_{i}$. Removing the constant term in the twist potentials allows us to assume that without loss of 
generality the entries which become zero or blow up towards a nut are only on the diagonal. The off-diagonal entries are bounded towards the nut.

Without loss of generality assume that the nut is at $a=0$. Then the calculations above show that to leading order in $z$ the patching matrices below and above the nut are 
(Chosen the right orientation for the basis such that the signs which are recorded by $\epsilon$ work out.)
\begin{equation} \label{eq:Pplus} 
\renewcommand{\arraystretch}{1.5}
P_{-}^{\vphantom{\frac{1}{2}}}=
\left(\begin{array}{ccc}
\hphantom{-}\dfrac{1}{2zX_{0}U_{0}} & -\dfrac{Z_{0}}{(U_{0}X_{0}W_{0})^{\scriptscriptstyle{1/2}}_{\vphantom{0}}} & -\dfrac{V_{0} z}{(U_{0}X_{0}W_{0})^{\scriptscriptstyle{1/2}}_{\vphantom{0}}} \\
\cdot & X_{0} & 2zY_{0} \\
\cdot & \cdot & 2zU_{0}
\end{array} \right),
\end{equation}
\begin{equation} \label{eq:Pminus} 
\renewcommand{\arraystretch}{1.5}
P_{+}^{\vphantom{\frac{1}{2}}}=
\left(\begin{array}{ccc}
-\dfrac{1}{2zX_{0}W_{0}} & \hphantom{-}\dfrac{Y_{0}}{(U_{0}X_{0}W_{0})^{\scriptscriptstyle{1/2}}_{\vphantom{0}}} & -\dfrac{V_{0} z}{(U_{0}X_{0}W_{0})^{\scriptscriptstyle{1/2}}_{\vphantom{0}}} \\
\cdot & X_{0} & -2zZ_{0} \\
\cdot & \cdot & -2zW_{0}
\end{array} \right).
\end{equation}
Using that $\det J = -u^{2}v^{2}$ in \eqref{eq:Jaroundnut} and thus $X_{0}U_{0}W_{0}=-1$ to leading order in $z$ we see that the switching is correct to leading order 
in $z$. (This is consistent with the different adaptations we calculated for example for the Schwarzschild space-time or flat space, see Section~\ref{subsec:PMP}.)

\begin{proof}[Proof of Theorem~\ref{thm:switching}]
To prove Theorem~\ref{thm:switching} the strategy is to follow the splitting procedure outlined in \cite[Sec.~8.4]{Metzner:2012aa}.

We first observe that splitting $P_{+}$ as in \eqref{eq:Pplus} will lead not to $J(r,z)$ as desired, but to $J(r,z)$ with its rows and columns permuted. This can be 
seen by looking at the diagonal case. To obtain $J(r,z)$ with the rows and columns in the order $(t,\phi,\psi)$ we need to permute
\begin{equation*}
P_{+}\rightarrow \widetilde{P_{+}}=E_{1}P_{+}E_{1} \quad \text{with }
E_{1} = \left(\begin{array}{ccc}
0 & 1 & 0 \\
1 & 0 & 0 \\
0 & 0 & 1
\end{array} \right).
\end{equation*}
Similarly for $P_{-}$ by \eqref{eq:Pminus}, we permute
\begin{equation*}
 P_{-}\rightarrow\widetilde{P_{-}}=E_{2}P_{-}E_{2}^{\mathrm{t}} \quad \text{with }
E_{2} = \left(\begin{array}{ccc}
0 & 1 & 0 \\
0 & 0 & 1 \\
1 & 0 & 0
\end{array} \right).
\end{equation*}
 Note that now the prescription in Theorem~\ref{thm:switching} translates to
\begin{equation} \label{eq:tPplus}
\renewcommand{\arraystretch}{1.5}
\widetilde{P_{-}}=D\widetilde{P_{+}}D \quad \text{with }
D = \left(\begin{array}{ccc}
1 & 0 & 0 \\
0 & 2z & 0 \\
0 & 0 & \dfrac{1}{2z}
\end{array} \right).
\end{equation}
Recall that we have set $a=0$. Following the splitting procedure, to obtain $J$ we split the matrices 
\begin{equation} \label{eq:Pswitch1}
\renewcommand{\arraystretch}{1.5}
\begin{split}
\widehat{P_{+}} = 
\left(\begin{array}{ccc}
1 & 0 & 0 \\
0 & \dfrac{r}{ζ} & 0 \\
0 & 0 & 1
\end{array} \right)
\widetilde{P_{+}}
\left(\begin{array}{ccc}
1 & 0 & 0 \\
0 & -rζ & 0 \\
0 & 0 & 1
\end{array} \right),\\
\renewcommand{\arraystretch}{1.5}
\widehat{P_{-}} = 
\left(\begin{array}{ccc}
1 & 0 & 0 \\
0 & 1 & 0 \\
0 & 0 & \dfrac{r}{ζ}
\end{array} \right)
\widetilde{P_{-}}
\left(\begin{array}{ccl}
1 & 0 & \hphantom{l}0 \\
0 & 1 & \hphantom{l}0 \\
0 & 0 & -rζ
\end{array} \right).
\end{split}
\end{equation}
The location of the diagonal entries which are not one is dictated by the position of the Killing vector which vanishes on the section of axis under consideration 
within the basis of Killing vectors $(\partial_{t},\partial_{\phi},\partial_{\psi})$ . In the language of Section~5.3 in Part~I, the integers $(p_{0},p_{1},p_{2})$ 
are, as we know, a permutation of $(0,0,1)$ and the location of the 1 is determined by the prescription just given.

Assembling \eqref{eq:tPplus} and \eqref{eq:Pswitch1} to
\begin{equation*}
\widehat{P_{-}}=A\widehat{P_{+}}B,
\end{equation*}
where 
\begin{equation*}
\renewcommand{\arraystretch}{1.5}
A = 
\left(\begin{array}{ccc}
1 & 0 & 0 \\
0 & \dfrac{2zζ}{r} & 0 \\
0 & 0 & \dfrac{r}{2zζ}
\end{array} \right), 
\quad
B = 
\left(\begin{array}{ccc}
1 & 0 & 0 \\
0 & -\dfrac{2z}{rζ} & 0 \\
0 & 0 & -\dfrac{rζ}{2z}
\end{array} \right),
\end{equation*}
all that is needed for completing the proof is to show that splitting the left and right hand side of this last equation yield the same $J$-matrix. To perform the 
splitting, we replace all appearances of $z$ by $w$ and make the substitution as in Eq.~(3.4) in Part~I. Note that
\begin{equation*} 
w=z+\frac{r}{2}\left(\zeta^{-1}-\zeta\right)=\frac{1}{2}\left(u^{2}-v^{2}+uv\left(\zeta^{-1}-\zeta\right)\right)=\frac{1}{2\zeta}\left(u\zeta+v\right)\left(u-v\zeta\right),
\end{equation*}
so that
\begin{align*}
\frac{2w\zeta}{r} & = \frac{(u\zeta+v)(u-v\zeta)}{uv}=1+O\left(\zeta\right), \\
-\frac{2w}{r\zeta} & = -\frac{1}{\zeta^2}\frac{(u\zeta+v)(u-v\zeta)}{uv}=1+O\left(\zeta^{-1}\right).
\end{align*}
Thus $A(z,r,\zeta)$ is holomorphic and nonsingular in the neighbourhood of $\zeta=0$ with $A(z,r,0)=\mathrm{id}$, and $B(z,r,\zeta^{-1})$ is holomorphic and 
nonsingular in the neighbourhood of $\zeta^{-1}=0$ with $B(z,r,0)=\mathrm{id}$. Consequently, if $\widehat{P_{+}^{\vphantom{\infty}}}$ splits as
\begin{equation*}
\widehat{P_{+}^{\vphantom{\infty}}}=K_{+}^{0}\left(r,z,\zeta\right)\left(K_{+}^{\infty}\left(r,z,\zeta^{-1}\right)\right)^{-1},
\end{equation*}
with $K_{+}^{0}$ holomorphic and nonsingular in the neighbourhood of $\zeta=0$ and $K_{+}^{\infty}$ holomorphic and nonsingular in the neighbourhood of $\zeta^{-1}=0$, 
then a splitting of $\widehat{P_{-}^{\vphantom{\infty}}}$ 
is given by taking
\begin{equation*}
\widehat{P_{-}^{\vphantom{\infty}}}=K_{-}^{0}\left(K_{-}^{\infty}\right)^{-1} \quad \text{with } K_{-}^{0}=A K_{+}^{0},\ K_{-}^{\infty}=B^{-1}K_{+}^{\infty}.
\end{equation*}
The corresponding expressions for $J$ are
\begin{equation*}
J=J_{+}^{\vphantom{\infty}}(r,z)=K_{+}^{0}(0)\left(K_+^{\infty}(0)\right)^{-1}
\end{equation*}
and
\begin{equation*} 
J=J_{-}^{\vphantom{\infty}}(r,z)=K_{-}^{0}(0)\left(K_{-}^{\infty}(0)\right)^{-1}=A(r,z,0) J_{+}^{\vphantom{\infty}}(r,z) B(r,z,0)=J_{+}^{\vphantom{\infty}}(r,z).
\end{equation*}
These are the same.
\end{proof}

\subsection{Application to the Black Ring} \label{sec:applbr}
Now we see how to apply the prescription for switching and the discussion of conicality to $P(z)$ for the black ring as in \eqref{eq:BRpatmat}. We are interested in the 
nut with largest $z$-value, which is the one at $z=\beta$. The first step is to make an additive shift to the twist potential $\chi$ to ensure that the term $P_{12}$ in 
\eqref{eq:BRpatmat} is finite at $z=\beta$. This needs
\begin{equation*}
\chi\rightarrow\chi+C, \quad C=-\frac{2\nu}{\beta+\lambda},
\end{equation*}
when
\begin{equation*} 
P_{12}\rightarrow P_{12}-CP_{11}=P_{12}-\frac{\nu(z+\lambda)}{(\beta+\lambda)(z+\alpha)(z-\beta)}=\frac{\nu}{(z+\alpha)(\beta+\lambda)},
\end{equation*}
which is indeed finite at $z=\beta$, and
\begin{equation*} 
P_{22}\rightarrow P_{22}-2CP_{12}+C^2P_{11}=-\frac{(z+\mu)}{(z+\alpha)}, \quad \text{where }\mu= \frac{\kappa^2(2b-c+bc)}{(1+b)},
\end{equation*}
which is also finite at $z=\beta$. We are in position to make the switch as at Theorem~\ref{thm:switching} with $\beta$ in place of $a$ and the result is
\begin{equation*}
\renewcommand{\arraystretch}{2.5}
P_{2}=\left(\begin{array}{ccc}
\dfrac{(z+\alpha)}{2(z-\alpha)(z-\beta)} & 0 & 0 \\
\cdot & -\dfrac{(z+\mu)}{(z+\alpha)} & \dfrac{2\nu(z-\beta)}{\gamma(z+\alpha)} \\
\cdot & \cdot & -2\dfrac{(z+\lambda)(z-\beta)}{(z+\alpha)}
\end{array} \right).
\end{equation*}
We have completed the switching and obtained $P_2(z)$, the transition matrix adapted to the section of axis $\alpha<z<\beta$. We could continue to find the transition 
matrix adapted to the other segments but that is straightforward and we do not need it. Instead we shall return to the question of conicality addressed in 
Corollary~\ref{cor:conic}. Compare with Theorem~\ref{thm:Jaroundnut} to find from $P_1$ that
\begin{equation*}
v^{2}W_{0}= \frac{2(z+\alpha)(z-\beta)}{(z-\alpha)}
\end{equation*}
where now $v^{2}=-2(z-\beta)$ and from $P_2$ that
\begin{equation*}
u^{2}U_{0}=-\frac{2(z+\lambda)(z-\beta)}{(z+\alpha)}
\end{equation*}
where now $u^{2}=2(z-\beta)$. Corollary~\ref{cor:conic} implies that there is no conical singularity on the axis section $\alpha<z<\beta$ provided
\begin{equation*}
\lim_{u\rightarrow 0}W_{0}=\lim_{v\rightarrow 0}U_{0},
\end{equation*}
which here requires
\begin{equation*}
\frac{\beta+\lambda}{\beta+\alpha}=\frac{\beta+\alpha}{\beta-\alpha}.
\end{equation*}
Using \eqref{eq:BRpar} this condition can be solved for $b$ as
\begin{equation*}
b=\frac{2c}{1+c^2}
\end{equation*}
which is known  to be the right condition (\cite{Emparan:2002aa} or \cite[Eq.~(6.20)]{Harmark:2004rm}). 

%% file: figure1.tex
\begin{picture}(0,0)%
\epsfig{file=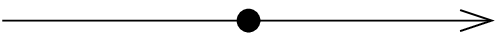}%
\end{picture}%
\setlength{\unitlength}{4144sp}%
\begingroup\makeatletter\ifx\SetFigFont\undefined%
\gdef\SetFigFont#1#2#3#4#5{%
  \reset@font\fontsize{#1}{#2pt}%
  \fontfamily{#3}\fontseries{#4}\fontshape{#5}%
  \selectfont}%
\fi\endgroup%
\begin{picture}(2464,1060)(4039,-1521)
\put(6401,-1236){\makebox(0,0)[lb]{\smash{{\SetFigFont{12}{14.4}{\rmdefault}{\mddefault}{\updefault}{\color[rgb]{0,0,0}$z$}%
}}}}
\put(5513,-861){\makebox(0,0)[lb]{\smash{{\SetFigFont{12}{14.4}{\rmdefault}{\mddefault}{\updefault}{\color[rgb]{0,0,0}$\frac{\partial}{\partial \varphi}=0$}%
}}}}
\put(5001,-1461){\makebox(0,0)[lb]{\smash{{\SetFigFont{12}{14.4}{\rmdefault}{\mddefault}{\updefault}{\color[rgb]{0,0,0}$z=0$}%
}}}}
\put(4451,-1461){\makebox(0,0)[lb]{\smash{{\SetFigFont{12}{14.4}{\rmdefault}{\mddefault}{\updefault}{\color[rgb]{0,0,0}2}%
}}}}
\put(5769,-1461){\makebox(0,0)[lb]{\smash{{\SetFigFont{12}{14.4}{\rmdefault}{\mddefault}{\updefault}{\color[rgb]{0,0,0}1}%
}}}}
\put(4351,-861){\makebox(0,0)[lb]{\smash{{\SetFigFont{12}{14.4}{\rmdefault}{\mddefault}{\updefault}{\color[rgb]{0,0,0}$\frac{\partial}{\partial \psi}=0$}%
}}}}
\end{picture}%

%% file: figure2.tex
\begin{picture}(0,0)%
\epsfig{file=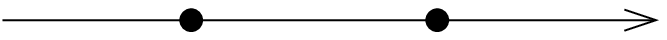}%
\end{picture}%
\setlength{\unitlength}{4144sp}%
\begingroup\makeatletter\ifx\SetFigFont\undefined%
\gdef\SetFigFont#1#2#3#4#5{%
  \reset@font\fontsize{#1}{#2pt}%
  \fontfamily{#3}\fontseries{#4}\fontshape{#5}%
  \selectfont}%
\fi\endgroup%
\begin{picture}(3214,1110)(3289,-4471)
\put(3488,-3761){\makebox(0,0)[lb]{\smash{{\SetFigFont{12}{14.4}{\rmdefault}{\mddefault}{\updefault}{\color[rgb]{0,0,0}$\frac{\partial}{\partial \psi}=0$}%
}}}}
\put(4651,-3761){\makebox(0,0)[lb]{\smash{{\SetFigFont{12}{14.4}{\rmdefault}{\mddefault}{\updefault}{\color[rgb]{0,0,0}$\mathcal{H}$}%
}}}}
\put(4651,-4411){\makebox(0,0)[lb]{\smash{{\SetFigFont{12}{14.4}{\rmdefault}{\mddefault}{\updefault}{\color[rgb]{0,0,0}2}%
}}}}
\put(4026,-4411){\makebox(0,0)[lb]{\smash{{\SetFigFont{12}{14.4}{\rmdefault}{\mddefault}{\updefault}{\color[rgb]{0,0,0}$-\alpha$}%
}}}}
\put(5101,-4411){\makebox(0,0)[lb]{\smash{{\SetFigFont{12}{14.4}{\rmdefault}{\mddefault}{\updefault}{\color[rgb]{0,0,0}$+\alpha$}%
}}}}
\put(3651,-4411){\makebox(0,0)[lb]{\smash{{\SetFigFont{12}{14.4}{\rmdefault}{\mddefault}{\updefault}{\color[rgb]{0,0,0}3}%
}}}}
\put(5751,-4411){\makebox(0,0)[lb]{\smash{{\SetFigFont{12}{14.4}{\rmdefault}{\mddefault}{\updefault}{\color[rgb]{0,0,0}1}%
}}}}
\put(6401,-4161){\makebox(0,0)[lb]{\smash{{\SetFigFont{12}{14.4}{\rmdefault}{\mddefault}{\updefault}{\color[rgb]{0,0,0}$z$}%
}}}}
\put(5551,-3761){\makebox(0,0)[lb]{\smash{{\SetFigFont{12}{14.4}{\rmdefault}{\mddefault}{\updefault}{\color[rgb]{0,0,0}$\frac{\partial}{\partial \varphi}=0$}%
}}}}
\end{picture}%

%% file: figure3.tex
\begin{picture}(0,0)%
\epsfig{file=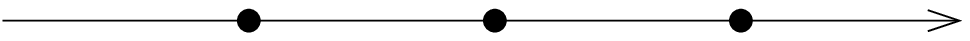}%
\end{picture}%
\setlength{\unitlength}{4144sp}%
\begingroup\makeatletter\ifx\SetFigFont\undefined%
\gdef\SetFigFont#1#2#3#4#5{%
  \reset@font\fontsize{#1}{#2pt}%
  \fontfamily{#3}\fontseries{#4}\fontshape{#5}%
  \selectfont}%
\fi\endgroup%
\begin{picture}(4589,1060)(889,-3771)
\put(2501,-3711){\makebox(0,0)[lb]{\smash{{\SetFigFont{12}{14.4}{\rmdefault}{\mddefault}{\updefault}{\color[rgb]{0,0,0}3}%
}}}}
\put(2501,-3111){\makebox(0,0)[lb]{\smash{{\SetFigFont{12}{14.4}{\rmdefault}{\mddefault}{\updefault}{\color[rgb]{0,0,0}$\mathcal{H}$}%
}}}}
\put(4526,-3111){\makebox(0,0)[lb]{\smash{{\SetFigFont{12}{14.4}{\rmdefault}{\mddefault}{\updefault}{\color[rgb]{0,0,0}$\frac{\partial}{\partial \varphi}=0$}%
}}}}
\put(3476,-3111){\makebox(0,0)[lb]{\smash{{\SetFigFont{12}{14.4}{\rmdefault}{\mddefault}{\updefault}{\color[rgb]{0,0,0}$\frac{\partial}{\partial \psi}=0$}%
}}}}
\put(1201,-3111){\makebox(0,0)[lb]{\smash{{\SetFigFont{12}{14.4}{\rmdefault}{\mddefault}{\updefault}{\color[rgb]{0,0,0}$\frac{\partial}{\partial \psi}=0$}%
}}}}
\put(5376,-3486){\makebox(0,0)[lb]{\smash{{\SetFigFont{12}{14.4}{\rmdefault}{\mddefault}{\updefault}{\color[rgb]{0,0,0}$z$}%
}}}}
\put(4701,-3711){\makebox(0,0)[lb]{\smash{{\SetFigFont{12}{14.4}{\rmdefault}{\mddefault}{\updefault}{\color[rgb]{0,0,0}1}%
}}}}
\put(3663,-3711){\makebox(0,0)[lb]{\smash{{\SetFigFont{12}{14.4}{\rmdefault}{\mddefault}{\updefault}{\color[rgb]{0,0,0}2}%
}}}}
\put(1251,-3711){\makebox(0,0)[lb]{\smash{{\SetFigFont{12}{14.4}{\rmdefault}{\mddefault}{\updefault}{\color[rgb]{0,0,0}4}%
}}}}
\put(4101,-3711){\makebox(0,0)[lb]{\smash{{\SetFigFont{12}{14.4}{\rmdefault}{\mddefault}{\updefault}{\color[rgb]{0,0,0}$+\beta$}%
}}}}
\put(1851,-3711){\makebox(0,0)[lb]{\smash{{\SetFigFont{12}{14.4}{\rmdefault}{\mddefault}{\updefault}{\color[rgb]{0,0,0}$-\alpha$}%
}}}}
\put(2976,-3711){\makebox(0,0)[lb]{\smash{{\SetFigFont{12}{14.4}{\rmdefault}{\mddefault}{\updefault}{\color[rgb]{0,0,0}$+\alpha$}%
}}}}
\end{picture}%

%% file: figure4.tex
\begin{picture}(0,0)%
\epsfig{file=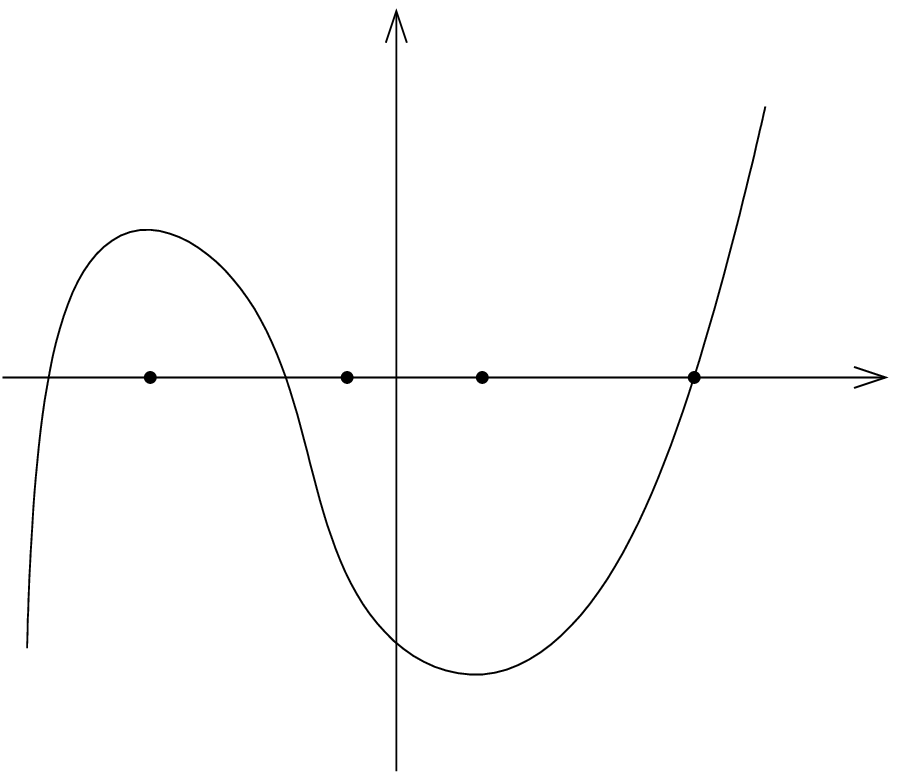}%
\end{picture}%
\setlength{\unitlength}{4144sp}%
\begingroup\makeatletter\ifx\SetFigFont\undefined%
\gdef\SetFigFont#1#2#3#4#5{%
  \reset@font\fontsize{#1}{#2pt}%
  \fontfamily{#3}\fontseries{#4}\fontshape{#5}%
  \selectfont}%
\fi\endgroup%
\begin{picture}(4214,3747)(3589,-5923)
\put(5226,-2323){\makebox(0,0)[lb]{\smash{{\SetFigFont{11}{13.2}{\rmdefault}{\mddefault}{\updefault}{\color[rgb]{0,0,0}$F(a)$}%
}}}}
\put(7701,-4148){\makebox(0,0)[lb]{\smash{{\SetFigFont{11}{13.2}{\rmdefault}{\mddefault}{\updefault}{\color[rgb]{0,0,0}$a$}%
}}}}
\put(6688,-4336){\makebox(0,0)[lb]{\smash{{\SetFigFont{11}{13.2}{\rmdefault}{\mddefault}{\updefault}{\color[rgb]{0,0,0}$A$}%
}}}}
\put(5676,-4311){\makebox(0,0)[lb]{\smash{{\SetFigFont{11}{13.2}{\rmdefault}{\mddefault}{\updefault}{\color[rgb]{0,0,0}$a_{\mathrm{min}}$}%
}}}}
\put(5026,-4311){\makebox(0,0)[lb]{\smash{{\SetFigFont{11}{13.2}{\rmdefault}{\mddefault}{\updefault}{\color[rgb]{0,0,0}$-\alpha$}%
}}}}
\put(4126,-4311){\makebox(0,0)[lb]{\smash{{\SetFigFont{11}{13.2}{\rmdefault}{\mddefault}{\updefault}{\color[rgb]{0,0,0}$a_{\mathrm{max\vphantom{i}}}$}%
}}}}
\end{picture}%

%% file: figure5.tex
\begin{picture}(0,0)%
\epsfig{file=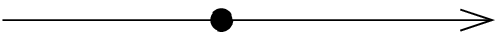}%
\end{picture}%
\setlength{\unitlength}{4144sp}%
\begingroup\makeatletter\ifx\SetFigFont\undefined%
\gdef\SetFigFont#1#2#3#4#5{%
  \reset@font\fontsize{#1}{#2pt}%
  \fontfamily{#3}\fontseries{#4}\fontshape{#5}%
  \selectfont}%
\fi\endgroup%
\begin{picture}(2389,1060)(3489,-2221)
\put(3701,-1561){\makebox(0,0)[lb]{\smash{{\SetFigFont{12}{14.4}{\rmdefault}{\mddefault}{\updefault}{\color[rgb]{0,0,0}$\frac{\partial}{\partial \psi}=0$}%
}}}}
\put(4801,-1561){\makebox(0,0)[lb]{\smash{{\SetFigFont{12}{14.4}{\rmdefault}{\mddefault}{\updefault}{\color[rgb]{0,0,0}$\frac{\partial}{\partial \varphi}=0$}%
}}}}
\put(5776,-1911){\makebox(0,0)[lb]{\smash{{\SetFigFont{12}{14.4}{\rmdefault}{\mddefault}{\updefault}{\color[rgb]{0,0,0}$z$}%
}}}}
\put(4801,-2161){\makebox(0,0)[lb]{\smash{{\SetFigFont{12}{14.4}{\rmdefault}{\mddefault}{\updefault}{\color[rgb]{0,0,0}$u=0$}%
}}}}
\put(3851,-2161){\makebox(0,0)[lb]{\smash{{\SetFigFont{12}{14.4}{\rmdefault}{\mddefault}{\updefault}{\color[rgb]{0,0,0}$v=0$}%
}}}}
\end{picture}%

%% file: Outlook.tex
\section{Summary and Outlook}
In this work we have presented a possible way for the reconstruction of five- or higher-dimensional black hole space-times from what are at the moment believed 
to be the classifying parameters, namely the rod structure and angular momenta. The method is based on a twistor construction which in turn relies on the Penrose-Ward transform. 

Our idea assigns a patching matrix to every rod structure where, apart from the possible poles at the nuts, the entries of the patching matrix have to be rational 
functions with the same denominator $Δ$ --- Section~\ref{sec:converse}. By imposing boundary conditions the aim is to determine all the coefficients of the polynomials 
in the numerator of these rational functions in terms of the nuts, rods and angular momenta.

However, with an increasing number of nuts one needs increasingly sophisticated tools and it is of particular importance to gain a detailed understanding of how the 
patching matrices, adapted to two neighbouring rods, are related. In Theorem~\ref{thm:switching} we show how to do this and \cite[Thm.~6.5]{Metzner:2012aa} provides 
this statement for the nut at infinity, that is, it relates the patching matrices which are adapted to the outer rods. By means of that we are able to reconstruct 
the patching matrix for a general two-nut rod structure and we can show that a 
three-nut rod structure with one Killing vector hypersurface-orthogonal, together with a given angular momentum, fixes the space-time to be the black ring. 

Also in Section~\ref{sec:converse} we discuss conical singularities on the axis and show how to obtain necessary and sufficient conditions for their removal. 
Applying this to the black ring we obtain the known relation between the parameters. In particular, this implies a relation between the rod structure and the 
asymptotic quantities for a non-singular solution known to exist.\\

\noindent Further questions which are interesting to pursue in this context are for example: 

Which rod structures are admissible? In other words, are there restrictions on the rod structures arising in nonsingular solutions? The example of flat space treated here shows that there 
are restrictions.

Can we construct a Lens space-time this way, that is a space-time whose horizon is connected and has the topology of a Lens space \cite[Prop.~2]{Hollands:2008fp}? 
We know what the corresponding rod structure looks like, but are we able to fix enough parameters and can we see whether the resulting patching matrix does give 
rise to a space-time without singularities? The latter question seems to be difficult to address as by the analytic continuation one can guarantee the existence 
of the solution with all its nice regularity properties only in a neighbourhood of the axis, but further away from the axis there might be so-called ``jumping 
lines'', where the mentioned triviality assumption of the bundle does not hold. 

How many dimensions does the moduli space for an $n$\,-nut rod structure have? Can we find upper and lower bounds on that depending on the imposed boundary 
conditions? This also does not seem to be an easy questions as most of the conditions we impose on the patching matrix are highly non-linear, for example the 
determinant condition. 

Which parts of the theory extend to yet higher dimensions? We have already pointed out along the way that some statements straight-forwardly generalize to more 
than five dimensions as well, but some others do not. A closer look at those points would certainly be interesting. 

Also stepping down a dimension leads to a question for which this set of tools might be appropriate. Are we able to disprove the existence of a regular double-Kerr 
solution in four dimensions by these methods? It is conceivable that for example the imposed compatibility requirements as one switches at the nuts lead finally to an 
overdetermined system of conditions and thereby a contradiction.